\documentclass[a4paper,USenglish,cleveref, autoref, thm-restate,nolineno]{socg-lipics-v2021}

\usepackage{tikz}
\usepackage{todonotes}
\usepackage[vlined, ruled]{algorithm2e}

\graphicspath{{./figures/}}

\title{Dynamic Convex Hulls for Simple Paths\thanks{A preliminary version of this paper will appear in {\em Proceedings of the 40th International Symposium on Computational Geometry (SoCG 2024)}.}}

\titlerunning{Dynamic Convex Hulls for Simple Paths}

\author{Bruce Brewer}{Kahlert School of Computing,
University of Utah, Salt Lake City, UT 84112, USA}{bruce.brewer@utah.edu}{0009-0008-2995-148X}{Supported in part by NSF under Grant CCF-2300356.}

\author{Gerth Stølting Brodal}{Department of Computer Science,
Aarhus University,
Aabogade 34, 8200 Aarhus N, Denmark.}{gerth@cs.au.dk}{0000-0001-9054-915X}{Supported by Independent Research Fund Denmark, grant~9131-00113B.}

\author{Haitao Wang}{Kahlert School of Computing,
University of Utah, Salt Lake City, UT 84112, USA}{haitao.wang@utah.edu}{0000-0001-8134-7409}{Supported in part by NSF under Grant CCF-2300356.}

\authorrunning{B. Brewer, G.\,S. Brodal, and H. Wang}

\Copyright{Bruce Brewer, Gerth Stølting Brodal, Haitao Wang}



\relatedversion{} 

\acknowledgements{}

\usepackage{cite}
\usepackage{hyperref}
\hypersetup{
    colorlinks=true,
    linkcolor=red,
    filecolor=magenta,
    citecolor = blue,
    urlcolor=cyan,
    linktocpage,
    }

\hideLIPIcs
\EventEditors{John Q. Open and Joan R. Access}
\EventNoEds{2}
\EventLongTitle{42nd Conference on Very Important Topics (CVIT 2016)}
\EventShortTitle{CVIT 2016}
\EventAcronym{CVIT}
\EventYear{2016}
\EventDate{December 24--27, 2016}
\EventLocation{Little Whinging, United Kingdom}
\EventLogo{}
\SeriesVolume{42}
\ArticleNo{23}

\def\ch{\mathcal{H}}

\def\bbR{\mathbb{R}}
\def\IL{\textsc{Insert\-Left}}
\def\IR{\textsc{Insert\-Right}}
\def\DL{\textsc{Delete\-Left}}
\def\DR{\textsc{Delete\-Right}}
\def\IF{\textsc{Insert\-Front}}
\def\IRR{\textsc{Insert\-Rear}}
\def\DF{\textsc{Delete\-Front}}
\def\DRR{\textsc{Delete\-Rear}}
\def\TR{\textsc{Tree\-Retrieval}}
\def\HTR{\textsc{Hull\-Tree\-Retrieval}}
\def\HR{\textsc{Hull\-Report}}
\def\SQ{\textsc{Standard\-Query}}
\def\CON{\textsc{Concatenate}}

\let\epsilon\varepsilon

\begin{document}

\maketitle


\begin{abstract}
    We consider the planar dynamic convex hull problem. In the literature, solutions exist supporting the insertion and deletion of points in poly-logarithmic time and various queries on the convex hull of the current set of points in logarithmic time. If arbitrary insertion and deletion of points are allowed, constant time updates and fast queries are known to be impossible. This paper considers two restricted cases where worst-case constant time updates and logarithmic time queries are possible. We assume all updates are performed on a deque (double-ended queue) of points. The first case considers the monotonic path case, where all points are sorted in a given direction, say horizontally left-to-right, and only the leftmost and rightmost points can be inserted and deleted. The second case assumes that the points in the deque constitute a simple path. Note that the monotone case is a special case of the simple path case. For both cases, we present solutions supporting deque insertions and deletions in worst-case constant time and standard queries on the convex hull of the points in $O(\log n)$ time, where $n$ is the number of points in the current point set.  The convex hull of the current point set can be reported in $O(h+\log n)$ time, where $h$ is the number of edges of the convex hull. For the 1-sided monotone path case, where updates are only allowed on one side, the reporting time can be reduced to $O(h)$, and queries on the convex hull are supported in $O(\log h)$ time. All our time bounds are worst case. In addition, we prove lower bounds that match these time bounds, and thus our results are optimal. For a quick comparison, the previous best update bounds for the simple path problem were amortized $O(\log n)$ time by Friedman, Hershberger, and Snoeyink [SoCG 1989].
\end{abstract}

\section{Introduction}
\label{sec:intro}

Computing the convex hull of a set of $n$ points in the plane is a classic problem in computational geometry. Several algorithms can compute the convex hull in $O(n\log n)$ time in the static setting. For example, Graham's scan~\cite{Graham72} and Andrew's vertical sweep~\cite{Andrew79}. Andrew's algorithm can construct the convex hull in $O(n)$ time if the points are already sorted by either the $x$-coordinates or the $y$-coordinates. 
This has been generalized by Graham and Yao~\cite{ref:GrahamFi83} and by Melkman~\cite{Melkman} to construct the convex hull of a simple path in $O(n)$ time.
Output-sensitive algorithms of $O(n\log h)$ time have also been achieved by Kirkpatrick and Seidel~\cite{KirkpatrickSeidel86} and by Chan~\cite{Chan96}, where $h$ is the size of the convex hull.

Overmars and van Leeuwen~\cite{Overmars} studied the problem in the dynamic context where points can be inserted and deleted. Since a single point insertion and deletion can imply a linear change in the number of points on the convex hull, it is not desirable to report the entire convex hull explicitly after each update. Instead, one maintains a representation of the convex hull that can be queried. Overmars and van Leeuwen support the insertion and deletion of points in $O(\log^2 n)$ time, where $n$ is the number of points stored. They maintain the points in sorted order in one dimension as the elements of a binary tree and bottom-up maintain a hierarchical composition of the convex hull. Since the convex hull is maintained explicitly as essentially a linked list at the root, the $h$ points on the hull can be reported in $O(h)$ time, and queries on the convex hull can be supported in $O(\log n)$ time using appropriate binary searches. Some example convex hull queries are (see Figure~\ref{fig:convex-hull-queries}): Determine whether a point~$q$ is outside the convex hull, and if yes, compute the tangents (i.e., find the tangent points) of the convex hull through~$q$. Given a direction~$\rho$, compute an extreme point on the convex hull along~$\rho$. Given a line~$\ell$, determine whether $\ell$ intersects the convex hull, and if yes, find the two edges (bridges) on the convex hull intersected by~$\ell$. Tangent and extreme point queries are examples of \emph{decomposable} queries, which are queries whose answers can be obtained in constant time from the query answers for any constant number of subsets that form a partition of the point set. In contrast, bridge queries are not decomposable.

\begin{figure}[t]
    \centering
    \small
    \begin{tikzpicture}[scale=1.0]
        \newcommand{\dottedsize}{1pt}
        \draw[dashed, fill=black!10, line width=0.5pt] plot coordinates
            {(1,1) (2,3) (3,3.5) (5,4.25) (7,4.5) (10,3) (9.25,1) (8,0) (6,-1) (2,0) } -- cycle;  
        \draw[line width=0.5pt, solid, mark=*, mark size=1.75pt] plot coordinates
            {(2.5, 1.25) (2, 2) (1,1) (3, 0.5) (4, 1.5) (2,3) (3,3.5) (6, 3) (5,4.25) (8, 3) (7,4.5) (10,3) (8, 1.5) (5.5, 2.5) (4, 0) (2,0) (6,-1) (5, 0) (8,0) (9.25,1) (6.5, 0.5) (6, 1.5)
            };  
        \node[anchor=west] at (2.5,1.25) {$p_1$};
        \node[anchor=east] at (6,1.5) {$p_n$};
        \node[anchor=north east] at (2,0) {$t_1$};
        \node[anchor=south east] at (2,3) {$t_2$};
        \node[anchor=north] at (7,-0.6) {$e_1$};
        \node[anchor=south east] at (6,4.35) {$e_2$};
        \node[anchor=north west] at (9.25,1) {$p^\rho$};
        \draw[dotted, line width=\dottedsize] (8.25,-0.5) -- (10.25,2.5);  
        \draw[solid, line width=0.5pt, -latex] (9.75,1.75) -- +(1.5,-1) node[above, midway] {$\rho$};
        \draw[dotted, line width=\dottedsize] (4.5,4.75) -- (-0.5,1.25) -- (4.5,-1.25);  
        \draw[black, fill=white] (-0.5,1.25) circle (2pt) node [left] {$q$};
        \draw[dotted, line width=\dottedsize] (6.25,5) -- +(1.5,-6) node [below] {$\ell$};
    \end{tikzpicture}
    \caption{The convex hull (dashed) of a simple path $p_1,\ldots,p_n$ (solid).
        Three types of convex hull queries are shown (dotted): 
        the tangent points $t_1$ and $t_2$ with a query point $q$ outside the convex~hull;
        the extreme point $p^\rho$  in direction $\rho$;
        and the two convex hull edges $e_1$ and $e_2$ intersecting a line~$\ell$.}
    \label{fig:convex-hull-queries}
\end{figure}
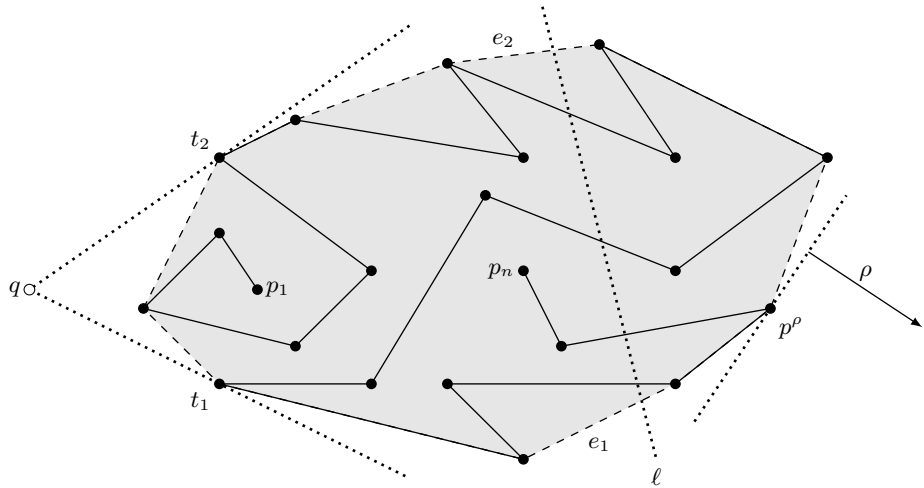

Chan~\cite{Chan01} improved the update (insertion/deletion) time to amortized $O(\log^{1+\epsilon}n)$, for any $\epsilon >0$, by not maintaining an explicit representation of the convex hull. Tangent and extreme point queries are supported in $O(\log n)$ time, allowing the $h$ points on the convex hull to be reported in $O(h\log n)$ time. The bridge query time was increased to $O(\log^{3/2} n)$. The update time was subsequently improved to amortized $O(\log n \log\log n)$ by Brodal and Jacob~\cite{BrodalJacob00} and Kaplan, Tarjan, Tsioutsiouliklis~\cite{KaplanTT01}, and to amortized $O(\log n)$ by Brodal and Jacob~\cite{BrodalJacob02}. Chan~\cite{Chan12} improved the time for bridge queries to $2^{O\big(\sqrt{\log\log n \log\log\log n}\big)}\log n$, with the same amortized update time. It is known that sub-logarithmic update time and logarithmic query time are not possible. For example, to achieve $O(\log n)$ time extreme point queries, an amortized update time $\Omega(\log n)$ is necessary~\cite{BrodalJacob00}.

In this paper, we consider the dynamic convex hull problem for restricted updates, where we can achieve worst-case constant update time and logarithmic query time. In particular, we assume that the points are inserted and deleted in a \emph{deque} (double-ended queue) and that they are geometrically restricted. We consider two restrictions: The first is the monotone path case, where all points in the deque are sorted in a given direction, say horizontally left-to-right, and only the leftmost and rightmost points can be inserted and deleted. The second case allows the points to form a simple path (i.e., a path that does not self intersections), where updates are restricted to both ends of the path. The simple path problem was previously studied by Friedman, Hershberger, and Snoeyink~\cite{FriedmanHS89}, who supported deque insertions in amortized $O(\log n)$ time, deletions in amortized $O(1)$ time, and queries in $O(\log n)$ time. Bus and Buzer~\cite{Bus15} considered a special case of the problem where insertions only happen to the ``front'' end of the path and deletions are only on points at the ``rear'' end. Based on the algorithm in~\cite{Melkman}, which can compute the convex hull of a simple path in linear time, they achieved $O(1)$ amortized update time to support $O(h)$ time hull reporting. However, hull queries were not considered in~\cite{Bus15}. Wang~\cite{Wang23} recently considered a special monotone path case where updates are restricted to queue-like updates, i.e., insert a point to the right of the point set and delete the leftmost point of the point set. Wang called it \emph{window-sliding} updates and achieved amortized constant time updates, hull queries in $O(\log h)$ time,\footnote{The runtime was $O(\log n)$ in the conference paper but was subsequently improved to $O(\log h)$ in the arXiv version \url{https://arxiv.org/abs/2305.08055}.} and hull reporting in $O(h)$ time.

\subsection{Our results}

We present dynamic convex hull data structures for both the monotone path and the simple path variants. For both problems, we support deque insertions and deletions in worst-case constant time. We can answer extreme point, tangent, and bridge queries in $O(\log n)$ time, and we can report the convex hull in $O(h+\log n)$ time. For the {\em one-sided} monotone case, where updates are only allowed on one side, the reporting time can be reduced to $O(h)$, and convex hull queries are supported in $O(\log h)$ time. That is, they are only dependent on the current hull size and independent of the number of points in the set. In addition, we show that these time bounds are the best possible by proving matching lower bounds. The previous and new bounds for the various restricted versions of the dynamic convex hull problem are summarized in Table~\ref{tab:results}.

Our results are obtained by a combination of several ideas. To support deque updates, we partition the deque into left and right parts and treat these parts as two independent stack problems. Queries then need to compose the convex hull information from both the stack problems. This strategy has previously been used by Friedman, Hershberger, and Snoeyink~\cite{FriedmanHS89} and by Wang \cite{Wang23}. To support deletions in the stack structures, we store rollback information when performing insertions. When one of the stacks becomes nearly empty, we repartition the deque into two new stacks of balanced sizes. To achieve worst-case bounds, the repartition is done with incremental global rebuilding ahead of time~\cite{Overmars83}. To achieve worst-case insertion time, we perform incremental merging of convex hull structures, where we exploit that the convex hulls of two horizontally separated sets can be combined in worst-case $O(\log n)$ time~\cite{Overmars} and that the convex hulls of a bipartition of a simple path can be combined in $O(\log^2 n)$ time~\cite{guibas1991compact}. To reduce the query bounds for the 1-sided monotone path problem to be dependent on $h$ instead of $n$, we adopt ideas from Sundar's priority queue with attrition~\cite{Sundar89}. In particular, we partition the stack of points into four lists (possibly with some interior ``redundant'' points removed), of which three lists are in convex position, and three lists have size $O(h)$. We believe this idea is interesting in its own right as, to our knowledge, this is the first time Sundar's approach has been used to solve a geometric problem.

\begin{table}[t]
    \caption{Known and new results for dynamic convex hull on paths.
        $O_A$ are amortized time bounds.
        -- denotes operation is not supported.
        For an update, $h$ denotes the maximum size of the hull before and after the update.
        DL = delete left, IR = insert right, etc.}
    \label{tab:results}
    \centering
    \tabcolsep1ex%
    {%
        \begin{tabular}{lcccccc}
            Reference                                        & DL       & IL            & IR            & DR          & Queries     & Reporting       \\
            \hline
            \multicolumn{7}{l}{\textbf{No geometric restrictions}}                                                                                    \\
            Preparata~\cite{Preparata79} + rollback          & --       & --            & $O(\log h)$   & $O(\log h)$ & $O(\log h)$ & $O(h)$          \\
            \multicolumn{7}{l}{\textbf{Monotone path}}                                                                                                \\
            Andrews' sweep~\cite{Andrew79}                   & --       & $O_A(1)$      & $O_A(1)$      & --          & $O(\log h)$ & $O(h)$          \\
            Wang~\cite{Wang23}                               & $O_A(1)$ & --            & $O_A(1)$      & --          & $O(\log h)$ & $O(h)$          \\
            \emph{New (Theorem~\ref{thm:deque-convex-hull})} & $O(1)$   & $O(1)$        & $O(1)$        & $O(1)$      & $O(\log n)$ & $O(h + \log n)$ \\
            \emph{New (Theorem~\ref{thm:stack-convex-hull})} & --       & --            & $O(1)$        & $O(1)$      & $O(\log h)$ & $O(h)$          \\
            \multicolumn{7}{l}{\textbf{Simple path}}                                                                                                  \\
            Friedman et al.~\cite{FriedmanHS89}              & $O_A(1)$ & $O_A(\log n)$ & $O_A(\log n)$ & $O_A(1)$    & $O(\log n)$ & --              \\
            Bus and Buzer~\cite{Bus15}                       & $O_A(1)$ & --            & $O_A(1)$      & --          & --          & $O(h)$          \\
            \emph{New (Theorem~\ref{thm:dpch})}              & $O(1)$   & $O(1)$        & $O(1)$        & $O(1)$      & $O(\log n)$ & $O(h + \log n)$ \\
            \hline
        \end{tabular}}
\end{table}

\subsection{Other related work}

Andrew's static algorithm~\cite{Andrew79} is an incremental algorithm that explicitly maintains the convex hull of the points considered so far. It can add the next point to the right and left of the convex hull in amortized constant time.

Preparata~\cite{Preparata79} presented an insertion-only solution maintaining the convex hull in an AVL~tree~\cite{AVL62} that supports the insertion of an arbitrary point in $O(\log h)$ time, queries on the convex hull in $O(\log h)$ time, and reporting queries in $O(h)$ time. For the \emph{stack} version of the dynamic convex hull problem, where updates form a stack, a general technique to support deletions is by having a stack of rollback information, i.e., the changes performed by the insertions. The time bound for deletions will then match the time bound for insertions, provided that insertion bounds are worst-case. Applying this idea to~\cite{Preparata79}, we have a stack dynamic convex hull solution with $O(\log h)$ time updates. Note that these time bounds hold for arbitrary new points inserted without geometric restrictions. The only limitation is that updates form a stack.

Hershberger and Suri~\cite{HershbergerSuri91} considered the offline version of the dynamic convex hull problem, assuming the sequence of insertions and deletions is known in advance, supporting updates in amortized $O(\log n)$ time. Hershberger and Suri~\cite{HershbergerSuri92} also considered the semi-dynamic deletion-only version of the problem, supporting initial construction and a sequence of $n$ deletions in $O(n\log n)$ time.

Given a simple path of $n$ vertices, Guibas, Hershberger, and Snoeyink~\cite{guibas1991compact} considered the problem of processing the path into a data structure so that the convex hull of a query subpath (specified by its two ends) can be (implicitly) constructed to support queries on the convex hull. Using a compact interval tree, they gave a data structure of $O(n\log\log n)$ space with $O(\log n)$ query time. The space was recently improved to $O(n)$ by Wang~\cite{ref:WangAl20}.
There are also other problems in the literature regarding convex hulls for simple paths. For example, Hershberger and Snoeyink~\cite{ref:HershbergerCa98} considered the problem of maintaining convex hulls for a simple path under split operations at certain extreme points, which improves the previous work in~\cite{ref:DobkinAn93}.

\subparagraph{Notation.}
We define some notation that will be used throughout the paper. For any compact subset $R$ of the plane (e.g., $R$ is a set of points or a simple path), let $\ch(R)$ denote the convex hull of $R$ and let $|\ch(R)|$ denote the number of vertices of $\ch(R)$. We also use~$\partial R$ to denote the boundary of $R$.

For a dynamic set $P$ of points, we define the following five operations: \IR: Insert a point to $P$ that is to the right of all of the points of~$P$; \DR: Delete the rightmost point of~$P$; \IL: Insert a point to $P$ that is to the left of all the points of $P$; \DL: Delete the leftmost point of~$P$; \HR: Report the convex hull~$\ch(P)$ (i.e., output the vertices of $\ch(P)$ in cyclic order around $\ch(P)$). We also use \SQ\ to refer to standard queries on~$\ch(P)$. This includes all decomposable queries like extreme point queries and tangent queries. It also includes certain non-decomposable queries like bridge queries.  Other queries, such as deciding if a query point is inside $\ch(P)$, can be reduced to bridge queries. 

We define the operations for the dynamic simple path $\pi$ similarly.
For convenience, we call the two ends of $\pi$ the {\em rear end} and the {\em front end}, respectively. As such, instead of ``left'' and ``right'', we use ``rear'' and ``front'' in the names of the update operations. Therefore, we have the following four updates: \IF, \DF, \IRR, and \DRR, in addition to \HR\ and \SQ\ as above.

\subparagraph{Outline.} We present our algorithms for the monotone path problem in Section~\ref{sec:monotone} and for the simple path problem in Section~\ref{sec:path}. In Section~\ref{sec:lower-bounds}, we prove lower bounds that match our results in the previous two sections.

\section{The monotone path problem}
\label{sec:monotone}

In this section, we study the monotone path problem where updates occur only at the extremes in a given direction, say, the horizontal direction. That is, given a set of points~$P \subset \mathbb{R}^2$, we maintain the convex hull of $P$, denoted by $\ch(P)$, while points to the left and right of $P$ may be inserted to $P$ and the rightmost and leftmost points of $P$ may be deleted from~$P$. Throughout this section, we let $n$ denote the size of the current set $P$ and $h=|\ch(P)|$. For ease of exposition, we assume that no three points of $P$ are collinear.

If updates are allowed at both sides (resp., at one side), we denote it the \emph{two-sided} (resp. {\em one-sided}) problem. We call the structure for the two-sided problem the ``deque convex hull,'' where we use the standard abbreviation deque to denote a double-ended queue (according to Knuth~\cite[Section~2.2.1]{Knuth73}, E.\ J.\ Schweppe introduced the term deque). The one-sided problem's structure is called the ``stack convex hull''.

In what follows, we start with describing a ``stack tree'' in Section~\ref{sec:stacktree}, which will be used to develop a ``deque tree'' in Section~\ref{sec:dequetree}. We will utilize the deque tree to implement the deque convex hull in Section~\ref{sec:2sidemontone} for the two-sided problem. The deque tree, along with ideas from Sundar's priority queues with attrition~\cite{Sundar89}, will also be used for constructing the stack convex hull in Section~\ref{sec:1sidemonotone} for the one-sided problem.

\subsection{Stack tree}
\label{sec:stacktree}

Suppose $P$ is a set of $n$ points in $\bbR^2$ sorted from left to right. Consider the following operations on $P$ (assuming $P=\emptyset$ initially). (1) \IR; (2) \DR; (3) \TR: Return the root of a balanced binary search tree (BST) that stores all points of the current $P$ in the left-to-right order. We have the following lemma.

\begin{lemma} \label{lem:stack-tree}
    Let $P$ be an initially empty set of points in $\bbR^2$ sorted from left to right.
    There exists a ``Stack Tree'' data structure $ST(P)$ for $P$ supporting the following operations:
    \begin{enumerate}
        \item \IR: $O(1)$ time.
        \item \DR: $O(1)$ time.
        \item \TR: $O(\log n)$ time, where $n$ is the current size of $P$.
    \end{enumerate}
\end{lemma}

In what follows, we describe the tree's structure and then discuss the operations. Unless otherwise stated, $P$ refers to the current point set, and $n=|P|$.

\subparagraph{Remark.}

It should be noted that the statement of Lemma~\ref{lem:stack-tree} is not new. Indeed, one can simply use a finger search tree~\cite{Brodal04,GuibasMPR77,ref:Tsakalidis85} to store $P$ to achieve the lemma (in fact, \TR\ can even be done in $O(1)$ time). We propose a stack tree as a new implementation for the lemma because it can be applied to our dynamic convex hull problem. When we use the stack tree, \TR\ will be used to return the root of a tree representing the convex hull of $P$; in contrast, simply using a finger search tree cannot achieve the goal  (the difficulty is how to efficiently maintain the convex hull to achieve constant time update). Our stack tree may be considered a framework for Lemma~\ref{lem:stack-tree} that potentially finds other applications as well.

\subparagraph{Structure of the stack tree.}
The structure of the stack tree is illustrated in Figure~\ref{fig:recursive-slowdown} to the right of the line $\ell$.
The stack tree $ST(P)$ consists of a sequence of trees $T_i$ for $i = 0, 1, \ldots, \lceil \log \log n \rceil$. Each $T_i$ is a balanced BST storing a contiguous subsequence of $P$ such that for any $j < i$, all points of $T_i$ are to the left of each point of $T_j$. The points of all $T_i$'s form a partition of $P$. We maintain the invariant that $|T_i|$ is a multiple of $2^{2^i}$ and $0 \leq |T_i| \leq 2^{2^{i + 1}}$, where $|T_i|$ represents the number of points stored in $T_i$. 
(The right side of $\ell$ in Figure~\ref{fig:recursive-slowdown} is a stack tree).

In order for our construction to achieve worst-case constant time insertions, the joining of two trees is performed incrementally over subsequent insertions. Specifically, we apply the {\em recursive slowdown} technique of Kaplan and Tarjan~\cite{KaplanTarjan95}, where every $2^{i+1}$-th insertion, $i \geq 1$, performs delayed incremental work toward joining $T_{i - 1}$ with $T_i$, if such a join is deemed necessary.

\begin{figure}[t]
    \centering
    \includegraphics[height=1.8in]{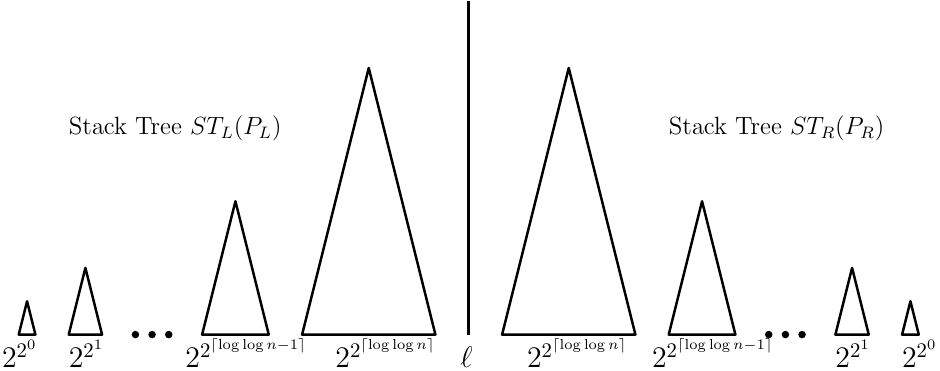}
    \caption{Illustrating a deque tree, comprising two stack trees separated by the vertical line $\ell$.}
    \label{fig:recursive-slowdown}
\end{figure}

\subparagraph{Remark.}
The critical observation of our algorithm is that because the ranges of the trees do not overlap, we can join adjacent trees $T_i$ and $T_{i + 1}$ to obtain (the root) of a new balanced~BST that stores all points in $T_i\cup T_{i + 1}$ in $O(\log (|T_i| + |T_{i + 1}|))$ time. Later in the paper we generalize this idea to horizontally neighboring convex hulls which can be merged in $O(\log(|\ch(T_i)| + |\ch(T_{i + 1})|))$ time \cite{Overmars} and to convex hulls over consecutive subpaths of a simple path which can be merged in $O(\log|\ch(T_i)| \cdot \log |\ch(T_{i + 1})|)$  time \cite{guibas1991compact}. In fact, the stack tree may be applied to solve other problems with polylogarithmic joining time so that that \TR\ (or other operations) can be performed efficiently. 

\subparagraph{\IR.}

Suppose we wish to insert into $P$ a point $p$ that is to the right of all points of~$P$.

We start with inserting $p$ into the tree $T_0$, which takes $O(1)$ time as $|T_0|=O(1)$. Next, we perform $O(1)$ delayed incremental work on a tree $T_i$ for a particular index $i$. To determine~$i$, we maintain a counter $N$ that is a binary number. Initially, $N=1$, and it is an invariant that $N=1+n$ (recall that $n$ is the size of the current point set $P$). For each insertion, we increment $N$ by one and determine the index $i$ of the digit which flips from $0$ to $1$, indexed from the right where the rightmost digit has index $0$. Note that there is exactly one such digit. Then, if $i\geq 1$, we perform incremental work on $T_i$ (i.e., joining $T_{i-1}$ with $T_i$). To find the digit $i$ in $O(1)$ time, we represent $N$ by a sequence of ranges, where each range represents a contiguous subsequence of digits of 1's in $N$. For example, if $N$ is 101100111, then the ranges are $[0, 2], [5, 6], [8, 8]$. After $N$ is incremented by one, $N$ becomes 101101000, and the ranges become $[3, 3], [5, 6], [8, 8]$. Therefore, based on the first two ranges in the range sequence, one can determine the digit that flips from $0$ to $1$ and update the range sequence in $O(1)$ time (note that this can be easily implemented using a linked list to store all ranges, without resorting to any bit tricks). 

After $i$ is determined, we perform incremental work on $T_i$ as follows. We use a variable~$n_j$ to maintain the size of each tree $T_j$, i.e., $n_j=|T_j|$. For each tree $T_j$, with $j\geq 1$, we say that $T_j$ is ``blocked'' if there is an incremental process for joining a previous $T_{j-1}$ with $T_j$ (more details to be given later) and ``unblocked'' otherwise ($T_0$ is always unblocked). If $T_i$ is blocked, then there is an incremental process for joining a previous $T_{i-1}$ with~$T_{i}$. This process will complete within time linear in the height of $T_i$, which is $O(2^i)$, since~$|T_i|\leq 2^{2^{i+1}}$. We perform the next $c$ steps for the process for a sufficiently large constant $c$. If the joining process is completed within the $c$ steps, we set $T_i$ to be unblocked.

Next, if $T_i$ is unblocked and $n_{i} \geq 2^{2^{i+1}}$ (in this case we prove in Observation~\ref{obser:treesize} that~$n_i$ must be exactly equal to $2^{2^{i+1}}$), then our algorithm maintains the invariant that $T_{i+1}$ must be unblocked, which will be proved in Lemma~\ref{lem:invariantstacktree}. In this case, we first set $T_{i+1}$ to be blocked, and then we start an incremental process to join $T_{i}$ with~$T_{i+1}$ without performing any actual steps. For reference purpose, let $T'_{i}$ refer to the current $T_{i}$ and let $T_{i}$ start over from $\emptyset$. Using this notation, we are actually joining~$T_{i}'$ with $T_{i+1}$. Although the joining process has not been completed, we follow the convention that $T'_{i}$ is now part of $T_{i+1}$; hence, we update~$n_{i+1}=n_{i+1}+n_{i}$. Also, since $T_{i}$ is now empty, we reset $n_{i}=0$. This finishes the work due to the insertion of $p$.

\begin{observation}\label{obser:treesize}
    \begin{enumerate}
        \item If $n_{i} \geq 2^{2^{i+1}}$, then $n_i=2^{2^{i+1}}$.
        \item It holds that $n_i=0$ or $2^{2^i}\leq n_i \leq 2^{2^{i+1}}$ for
              $i\geq 1$, and $n_0 \leq 4$.
    \end{enumerate}
\end{observation}
\begin{proof}
According to our algorithm, $T_0$ is never blocked and $n_0=4$ whenever $N=1+4k$ for $k=1,2,3,\ldots$, where $n_0$ is reset to $0$ (this is because whenever $n_0\geq 4$, $T_1$ must be unblocked, which is proved in Lemma~\ref{lem:invariantstacktree}). Hence, $n_0\leq 4$ always holds. In the following, we prove the case for $i\geq 1$.

Observe that starting from $n_i=0$, whenever $n_i$ increases, it increases by exactly the amount $2^{2^i}$. Hence, whenever $n_i\geq 2^{2^{i+1}}$ for the first time, it always holds that $n_i=2^{2^{i+1}}$. According to our algorithm, whenever $n_i\geq 2^{2^{i+1}}$, $n_i$ is reset to $0$ (this is because whenever $n_i\geq 2^{2^{i+1}}$, $T_{i+1}$ must be unblocked, which is proved in Lemma~\ref{lem:invariantstacktree}). The observation thus follows.
\end{proof}

The following lemma proves the algorithm invariant mentioned above.

\begin{lemma}\label{lem:invariantstacktree} \
    \begin{enumerate}
        \item
              If $n_0 \geq 4$, then $T_1$ must be unblocked.
        \item
              If $i\geq 1$ and $n_{i}\geq 2^{2^{i+1}}$ right after the process of joining $T_{i-1}$ with $T_{i}$ is completed, then $T_{i+1}$ must be unblocked.
    \end{enumerate}
\end{lemma}
\begin{proof}
    Recall that $N$ equals one plus the number of insertions performed, i.e., $N$ can be considered a timer.

    \subparagraph{Proof of the first lemma statement.}

    We have $n_0=4$ exactly when $N=1+4k$ for $k=1,2,3,\dots$, where $n_0$ is reset to 0, an incremental join at $T_1$ is initiated, and $T_1$ becomes blocked. But for the next insertion, i.e., when $N=2+4k$ for $k=1,2,3,\ldots$, we have always~$n_0=1$ and all the incremental $O(1)$ work at $T_1$ is processed if we make $c$ large enough, and $T_1$ becomes unblocked and remains unblocked until $n_0=4$ again. This proves the first lemma statement.

    \subparagraph{Proof of the second lemma statement.}

    We assume that the process of joining $T_{i-1}$ with~$T_{i}$ is completed due to the insertion of $p$ and $n_{i}\geq 2^{2^{i+1}}$.

    Let $N'$ be the most recent timer where $T_{i+1}$ changed status from unblocked to blocked. Hence, at $N'$, $T_{i}=\emptyset$, and the process of joining $T'_{i}$ with $T_{i+1}$ started at $N'$ and $T_{i+1}$ has been blocked due to that process. We argue below that by the time $N$, the incremental process for joining $T'_i$ with $T_{i+1}$ has already been completed.

    Indeed, the joining process takes $O(2^{i+1})$ time as the height of $T_{i+1}$ is $O(2^{i+1})$. Whenever the $(i+1)$-th bit of the counter $N$ flipped from $0$ to $1$, $c$ steps of the joining process will be performed. Note that the $(i+1)$-th bit of the counter $N$ flipped from $0$ to $1$ every $2^{i+2}$ insertions. Hence, for a sufficiently large constant $c$, the joining process will be completed within $2^{i+1}\cdot 2^{i+2}=2^{2i+3}$ insertions. If we make $c$ larger, then we can say that the joining process will be completed within $2^{2i}$ insertions.

    On the other hand, at $N'$, we have $T_{i}=\emptyset$. We argue that after $2^{2i}$ insertions, the size of~$T_{i}$ cannot be larger than or equal to $2^{2^{i+1}}$. Indeed, all points of $T_{i}$ are originally from $T_j$ for all $0\leq j < i$. At $N'$, the total number of points in all trees $T_j$ for $0\leq j < i$ is no larger than $2\cdot 2^{2^{i}}$ by Observation~\ref{obser:treesize}(2). Hence, after $2^{2i}$ insertions, the total number of points that can be inserted to $T_{i}$ is at most $2\cdot 2^{2^{i}}+2^{2i}$, which is smaller than $2^{2^{i+1}}$ for $i\geq 1$.

    The above implies that by the time $N$ when $n_i\geq 2^{2^{i+1}}$,  the incremental process for joining $T'_i$ with $T_{i+1}$ has already been completed, and therefore, $T_{i+1}$ cannot be blocked. This proves the second lemma statement. 
\end{proof}

As we only perform $O(1)$ incremental work, the total time for inserting $p$ is $O(1)$.

\subparagraph{\DR.}
To perform \DR, we maintain a stack that records the changes made on each insertion. To delete a point $p$, $p$ must be the most recently inserted point, and thus all changes made due to the insertion of $p$ are at the top of the stack. To perform the deletion, we simply pop the stack and roll back all the changes during the insertion of $p$.


\subparagraph{\TR.}
To perform \TR, we start by completing all incremental joining processes. Then, we join all trees $T_i$'s in their index order. This results in a single BST $T$ storing all points of $P$. In applications, we usually need to perform binary searches on $T$, after which we need to continue processing insertions and deletions on $P$. To this end, when constructing $T$ as above, we maintain a stack that records the changes we have made. Once we are done with queries on $T$, we use the stack to roll back the changes and return the stack tree to its original form right before the \TR\ operation.

For the time analysis, it takes $O(2^i)$ steps to finish an incremental joining process for each~$T_i$ (i.e., merging $T_{i-1}$ with $T_{i}$). Hence, it takes $O\left(\sum_{i = 1}^{\lceil \log \log n \rceil} 2^i \right) = O(\log n)$ time to finish all such processes. Next, joining all trees $T_i$'s in their index order takes $O\left(\sum_{i = 0}^{\lceil\log \log n\rceil} 2^i\right) = O(\log n)$ time in total. As such, the stack used to record changes made during the operation has size~$O(\log n)$ because it stores $O(\log n)$ changes. Rolling back all changes in the stack thus takes $O(\log n)$ time as well. Therefore, \TR\  can be performed in $O(\log n)$ time.

\subsection{Deque tree}
\label{sec:dequetree}

We now introduce the deque tree, which is built upon stack trees. 
We have the following lemma, where \TR\ is defined in the same way as in Section~\ref{sec:stacktree}.

\begin{lemma}
    \label{lem:deque-tree}
    Let $P$ be an initially empty set of points in $\bbR^2$ sorted from left to right.
    There exists a ``Deque Tree'' data structure $DT(P)$ for $P$ supporting the following operations:
    \begin{enumerate}
        \item \IR: $O(1)$ time.
        \item \DR: $O(1)$ time.
        \item \IL: $O(1)$ time.
        \item \DL: $O(1)$ time.
        \item \TR: $O(\log n)$ time, where $n$ is the size of the current set $P$.
    \end{enumerate}
\end{lemma}

As for Lemma~\ref{lem:stack-tree}, the statement of Lemma~\ref{lem:deque-tree} is not new because we can also use a finger search tree~\cite{Brodal04,GuibasMPR77} to achieve it. Here, we propose a different implementation for solving our dynamic convex hull problem.

The deque tree $DT(P)$ is built from two stack trees $ST_L(P_L)$ and $ST_R(P_R)$ from opposite directions, where $P_L$ and $P_R$ refer to the subsets of $P$ to the left and right of a vertical dividing line $\ell$, respectively (see Figure~\ref{fig:recursive-slowdown}). To insert a point to the left of $P$, we insert it to $ST_L(P_L)$. To delete the leftmost point of $P$, we delete it from $ST_L(P_L)$. For insertion/deletion on the right side of $P$, we use $ST_R(P_R)$. For \TR, we perform \TR\ operations on both $ST_L(P_L)$ and $ST_R(P_R)$, which obtain two balanced BSTs; then, we join these two trees into a single one. The time complexities of all these operations are as stated in the lemma.

To make this idea work, we need to make sure that neither $ST_L(P_L)$ nor $ST_R(P_R)$ is empty when $n$ is sufficiently large (e.g., if $n<c$ for a constant $c$, then we can do everything by brute force). To this end, we apply incremental global rebuilding~\cite[Section~5.2.2]{Overmars83}, where we dynamically adjust the dividing line $\ell$. Specifically, we maintain an invariant that $|P_L|$ and $|P_R|$ differ by at most a factor of 4, i.e., $\frac{1}{4} \leq \frac{|P_L|}{|P_R|} \leq 4$. To achieve this, whenever~$|P_L| = \frac{1}{2} |P_R|$ after an update, we set $\ell'$ to be the vertical line partitioning the current $P_R$ into two equal-sized subsets. Similarly, whenever $|P_R| = \frac{1}{2} |P_L|$, we set $\ell'$ to be the vertical line partitioning $P_L$ into two equal-sized subsets. Once $\ell'$ is set, we begin building $ST_L(P_L')$ on the subset $P_L'$ of points of $P$ to the left of $\ell'$ and $ST_R(P_R')$ on the subset $P_R'$ of points to the right of $\ell'$. It takes $O(n)$ time to build $ST_L(P_L')$ and $ST_R(P_R')$, and there will be $\Omega(n)$ updates (e.g., at least $\frac{n}{6}$ updates, where $n$ is the size of $|P|$ when the rebuilding procedure starts) between the time when $\frac{1}{2} \leq \frac{|P_L|}{|P_R|} \leq 2$ and when $\frac{1}{4} \leq \frac{|P_L|}{|P_R|} \leq 4$, so we can perform $O(1)$ incremental work to progressively build $ST_L(P_L')$ and $ST_R(P_R')$ for the next $O(n)$ updates (e.g., the next $\frac{n}{6}$ updates). By the time the above second condition is met, we may replace $ST_L(P_L)$ and $ST_R(P_R)$ with $ST_L(P_L')$ and $ST_R(P_R')$. This returns us to the state where $\frac{1}{2} \leq \frac{|P_L|}{|P_R|} \leq 2$. 


\subsection{Two-sided monotone path dynamic convex hull}
\label{sec:2sidemontone}

We can tackle the 2-sided monotone path dynamic convex hull problem using the deque tree. Suppose $P$ is a set of $n$ points in $\bbR^2$. In addition to the operations \IR, \DR, \IL, \DL, \HR, as defined in Section~\ref{sec:intro}, we also consider the operation \HTR: Return the root of a BST of height $O(\log h)$ that stores all vertices of the convex hull $\ch(P)$ (so that binary search based operations on~$\ch(P)$ can all be supported in $O(\log h)$ time). We will prove the following theorem.

\begin{theorem}
    \label{thm:deque-convex-hull}
    Let $P \subset \mathbb{R}^2$ be an initially empty set of points, with $n = |P|$ and $h = |\ch(P)|$.
    There exists a ``Deque Convex Hull'' data structure $DH(P)$ of $O(n)$ space that supports the following operations:
    \begin{enumerate}
        \item \IR: $O(1)$ time.
        \item \DR: $O(1)$ time.
        \item \IL: $O(1)$ time.
        \item \DL: $O(1)$ time.
        \item \HTR: $O(\log n)$ time.
        \item \HR: $O(h+\log n)$ time.
    \end{enumerate}
\end{theorem}

\subparagraph{Remark.}
The time complexities of the four update operations in Theorem~\ref{thm:deque-convex-hull} are obviously optimal. We will show in Section~\ref{sec:lower-bounds} that the other two operations are also optimal. In particular, it is not possible to reduce the time of \HTR\ to $O(\log h)$ or reduce the time of \HR\ to $O(h)$ (but this is possible for the one-sided case as shown in Section~\ref{sec:1sidemonotone}).
\medskip

The deque convex hull is a direct application of the deque tree from Section~\ref{sec:dequetree}. We maintain the upper hull and lower hull of $\ch(P)$ separately. In the following, we only discuss how to maintain the upper hull, as maintaining the lower hull is similar. By slightly abusing the notation, let $\ch(P)$ refer to the upper hull only in the following discussion.

We use a deque tree $DT(P)$ to maintain $\ch(P)$. The $DT(P)$ consists of two stack trees $ST_L$ and~$ST_R$. Each stack tree is composed of a sequence of balanced search trees $T_i$'s; each such tree $T_i$ stores left-to-right the points of the convex hull $\ch(P')$ for a contiguous subsequence~$P'$ of $P$. We follow the same algorithm as the deque tree with the following changes. During the process of joining $T_{i-1}$ with $T_{i}$, our task here becomes merging the two hulls stored in the two trees. To perform the merge, we first compute the upper tangent of the two hulls. This can be done in $O(\log (|T_{i-1}|+ |T_{i}|))$ time using the method of Overmars and van Leeuwen~\cite{Overmars}. Then, we split the tree $T_{i-1}$ into two portions at the tangent point; we do the same for $T_{i}$. Finally, we join the relevant portions of the two trees into a new tree that represents the merged hull of the two hulls. The entire procedure takes $O(\log (|T_{i-1}|+ |T_{i}|))$ time. This time complexity is asymptotically the same as joining two trees $T_{i-1}$ and $T_i$ as described in Section~\ref{sec:stacktree}, and thus we can still achieve the same performances for the first five operations as in Lemma~\ref{lem:deque-tree}; in particular, to perform \HTR, we simply call \TR\ on the deque tree. Finally, for \HR, we first perform \HTR\ to obtain a tree representing $\ch(P)$. Then, we perform an in-order traversal on the tree, which can output $\ch(P)$ in $O(h)$ time. Thus, the total time for \HR\ is $O(h+\log n)$.

\subsection{One-sided monotone path dynamic convex hull}
\label{sec:1sidemonotone}
We now consider the one-sided monotone problem. Suppose $P$ is a set of $n$ points in $\bbR^2$. Consider the following operations on $P$ (assume that $P=\emptyset$ initially): \IR, \DR, \HTR, \HR, as in Section~\ref{sec:2sidemontone}. Applying Theorem~\ref{thm:deque-convex-hull}, we can perform \HTR\ in $O(\log n)$ time and perform \HR\ in $O(h+\log n)$ time. We will prove the following theorem, which reduces the \HTR\ time to~$O(\log h)$ and reduces the \HR\ time to $O(h)$.

\begin{theorem}
    \label{thm:stack-convex-hull}
    Let $P \subset \mathbb{R}^2$ be an initially empty set of points, with $n = |P|$ and $h = |\ch(P)|$. There exists a ``Stack Convex Hull'' data structure $SH(P)$ of $O(n)$ space that supports the following operations:
    \begin{enumerate}
        \item \IR: $O(1)$ time.
        \item \DR: $O(1)$ time.
        \item \HTR: $O(\log h)$ time.
        \item \HR: $O(h)$ time.
    \end{enumerate}
\end{theorem}

The main idea to prove Theorem~\ref{thm:stack-convex-hull} is to adapt ideas from Sundar's algorithm in \cite{Sundar89} for priority queue with attrition as well as the deque convex hull data structure from Section~\ref{sec:2sidemontone}. As in Section~\ref{sec:1sidemonotone}, we maintain the upper and lower hulls of $\ch(P)$ separately. In the following, we only discuss how to maintain the upper hull. By slightly abusing the notation, let $\ch(P)$ refer to the upper hull only in the following discussion.

\subsubsection{Structure of the stack convex hull}
For any two disjoint subsets $P_1$ and $P_2$ of $P$, we use $P_1\prec P_2$ to denote the case where all points of $P_1$ are to the left of each point of $P_2$. Our data structure maintains four subsets  $A_1 \prec A_2 \prec A_3 \prec A_4$ of $P$. Each $A_i$, $1\leq i\leq 3$, is a convex chain, but this may not be true for $A_4$. Further, the following invariants are maintained during the algorithm.

\begin{enumerate}
    \item Vertices of $\ch(P)$ are all in $\bigcup_{i=1}^4 A_i$.
    \item $A_1$ is a prefix of the vertices of $\ch(P)$ sorted from left to right.
    \item $A_1 \cup A_2$ and $A_1 \cup A_3$ are both convex chains.
    \item BIAS: $|A_1| \geq |A_3| + 2 \cdot |A_4|$.
\end{enumerate}

The partition of the sequences of points into four subsequences and the above invariants are strongly inspired by Sundar's construction~\cite{Sundar89} for a priority queue supporting \textsc{Insert} and \textsc{DeleteMin} in worst-case constant time, and where insertions have the side-effect of deleting all elements larger than the inserted element from the priority queue.

For $1\leq i\leq 4$, we let $\textsc{Right}(A_i)$ and $\textsc{Left}(A_i)$ denote the rightmost and leftmost points of $A_i$, respectively, and let $\textsc{Left2}(A_i)$ and $\textsc{Right2}(A_i)$ denote the second leftmost and second rightmost point of $A_i$, respectively. For any three points $p_1 \prec p_2 \prec p_3$, we use $p_1 \to p_2 \to p_3$ to refer to the traversal from $p_1$ to $p_2$ and then to $p_3$; we often need to determine whether $p_1 \to p_2 \to p_3$ makes a left turn or right turn.

\subsubsection{\IR\ and \DR}
When we perform insertions, the BIAS invariant may be violated since $|A_4|$ increases by one (possibly after emptying $A_1$, $A_3$ and $A_4$). To restore it, we will perform the BIAS procedure (Algorithm~\ref{algo:bias}) twice.

\begin{algorithm}[t]
    \small
    \caption{\textsc{Bias}}
    \label{algo:bias}
    \KwIn{$A_1$, $A_2$, $A_3$, and $A_4$}
    \If{$|A_4| > 0$}{
        \If{$(|A_3| \geq 2$ and $\textsc{Right2}(A_3) \to \textsc{Right}(A_3) \to \textsc{Left}(A_4)$ is a left turn$)$
            or $(|A_3| = 1$ and $|A_1| \geq 1$ and $\textsc{Right}(A_1) \to \textsc{Right}(A_3) \to \textsc{Left}(A_4)$ is a left turn$)$
        }{
            $\textsc{DeleteRight}(A_3)$
        }
        \Else{
            $\textsc{InsertRight}(A_3, \textsc{DeleteLeft}(A_4))$\tcp*{Delete the leftmost point from $A_4$ and insert it to the right of $A_3$}
        }
    }
    \ElseIf{$|A_3| > 0$}{
        \If{$|A_2| > 0$ and $|A_3| \geq 2$ and $\textsc{Left}(A_2) \to \textsc{Left}(A_3) \to \textsc{Left2}(A_3)$ is a left turn}{
            $\textsc{DeleteLeft}(A_3)$
        }
        \ElseIf{$|A_2|>0$ and $(|A_1| = 0$ or $\textsc{Right}(A_1) \to \textsc{Left}(A_2) \to \textsc{Left}(A_3)$ is a right turn$)$}{
            $\textsc{InsertRight}(A_1, \textsc{DeleteLeft}(A_2))$
        }
        \Else{
            Set $A_2$ to empty \\
            $\textsc{InsertRight}(A_1, \textsc{DeleteLeft}(A_3))$
        }
    }
\end{algorithm}

One can verify that running the BIAS procedure once will increase $|A_1|-|A_3|-2\cdot|A_4|$ by at least one, unless $A_3=A_4=\emptyset$ (in which case the BIAS invariant trivially holds). Indeed, this is done by one of the following operations: (1) remove a point from $A_3$; (2) move a point from $A_4$ to $A_3$; (3) move a point from $A_2$ to $A_1$; (4) move a point from $A_3$ to $A_1$.

We wish to implement the BIAS procedure in $O(1)$ time. To this end, for each $A_i$, $1\leq i\leq 3$, since it is a convex chain and updates only happen at both ends of the chain, we store it by a finger search tree so that each such update can be supported in $O(1)$ time~\cite{Brodal04,GuibasMPR77,ref:Tsakalidis85}. For $A_4$, which may not be a convex chain, we store it by the deque convex hull data structure in Theorem~\ref{thm:deque-convex-hull}. In this way, the BIAS procedure runs in $O(1)$ time.

With the BIAS procedure in hand, we perform the \IR\ operation as shown in the pseudocode in Algorithm~\ref{algo:insert}. Due to the way we implement each $A_i$, \IR\ can be done in $O(1)$ time.

\begin{algorithm}[t]
    \small
    \caption{\IR}
    \label{algo:insert}
    \KwIn{
        $A_1$, $A_2$, $A_3$, $A_4$, and $p$, a point to be inserted}
    \If{$|A_1| \geq 2$ and $\textsc{Right2}(A_1) \to \textsc{Right}(A_1) \to p$ makes a left turn}{
        Set $A_2$ to refer to $A_1$ \\
        Set $A_1$, $A_3$, and $A_4$ to empty
    }
    \IR($A_4$, $p$) \\
    Run the \textsc{Bias} procedure twice \tcp{Restore BIAS invariant}
\end{algorithm}

To handle \DR, we use a stack to record changes made during each insertion. To delete a point $p$, we use the stack to roll back the changes done during the insertion of $p$ (these changes must be at the top of the stack). Since each \IR\ takes $O(1)$ time, there are $O(1)$ changes due to the insertion of $p$, and thus \DR\ can be accomplished in $O(1)$ time as well.

\subsubsection{\HTR}
We wish to (implicitly) construct a BST of height $O(\log h)$ to represent the upper hull $\ch(P)$. Recall that $h=|\ch(P)|$. The second and fourth invariants together imply that $|A_1|$, $|A_3|$, and~$|A_4|$ are all bounded by $O(h)$. Recall that for each $A_i$, $1\leq i\leq 3$, $A_i$ is a convex chain stored by a finger search tree, denoted by $T_i$. Since $|A_1|,|A_3|=O(h)$, the heights of both $T_1$ and $T_3$ are $O(\log h)$. For $A_4$, it is represented by a deque convex hull data structure. Since $|A_4|=O(h)$, by Theorem~\ref{thm:deque-convex-hull}, we can obtain in $O(\log h)$ time a tree $T_4$ of height $O(\log h)$ to store  $\ch(A_4)$. For~$A_2$, the only guaranteed upper bound for its size is $O(n)$. Hence, the height of $T_2$ is $O(\log n)$, instead of $O(\log h)$. By our algorithm invariants, $\ch(P)$ is the upper hull of $A_1 \cup A_2 \cup A_3 \cup \ch(A_4)$. As $A_1\prec A_2\prec A_3\prec A_4$, we can compute $\ch(P)$ by merging all four convex chains from left to right. This can be done in $O(\log n)$ time by using the trees $T_i$, $1\leq i\leq 4$, and the binary search method of Overmars and van Leeuwen~\cite{Overmars} for finding the upper tangents.

We now reduce the time to $O(\log h)$. The $O(\log n)$ time of the above algorithm is due to the fact that the height of $T_2$ is $O(\log n)$ instead of $O(\log h)$. This is because our algorithm invariants do not guarantee $|A_2|=O(h)$. Therefore, to have an $O(\log h)$ time algorithm, we need a clever way to merge $A_2$ with the other convex chains.

We first merge $T_3$ and $T_4$ to obtain  in $O(\log h)$ time a new tree $T_{34}$ of height $O(\log h)$ to represent the upper hull $\ch(A_3\cup A_4)$. Since $A_1$ is part of $\ch(P)$ by our algorithm invariant, we have the following two cases for $\ch(P)$: (1) no point of $A_2$ is on $\partial \ch(P)$; (2) at least one point of $A_2$ is on $\partial \ch(P)$. In the first case, $\ch(P)$ can be obtained by simply merging $T_1$ and~$T_{34}$ in $O(\log h)$ time. In the second case, vertices $\ch(P)$ from left to right are: all points of $A_1$, points of $A_2$ from $\textsc{Left}(A_2)$ to $p_2$, and points of $\ch(A_3\cup A_4)$ from $p_3$ to the right end, where $\overline{p_2p_3}$ is the upper tangent between $A_2$ and $\ch(A_3\cup A_4)$, with $p_2\in A_2$ and $p_3\in \partial\ch(A_3\cup A_4)$. As such, to compute $\ch(P)$ in the second case, the key is to compute the tangent $\overline{p_2p_3}$. We show below that this can be done in $O(\log h)$ time by an ``exponential search'' on $A_2$ and using the method of Overmars and van Leeuwen~\cite{Overmars}, in the following referred to as the OvL algorithm.

Note that which of the above two cases happens can be determined in $O(\log h)$ time. Indeed, we can first merge $T_1$ and $T_{34}$ to obtain a tree $T_{134}$ of height $O(\log h)$ representing the upper hull $\ch(P\setminus A_2)$. Then, due to our algorithm invariants, the first case happens if and only if $\textsc{Left}(A_2)$ is below $\ch(P\setminus A_2)$, which can be determined in $O(\log h)$ time using the tree $T_{134}$.

We now describe an algorithm to compute the common tangent $\overline{p_2p_3}$ in the second case. Let $\ch_{34}=\ch(A_3\cup A_4)$. Applied to the OvL algorithm directly using $T_2$ and $T_{34}$, it will take $O(\log (|A_2| + |A_3\cup A_4|)) = O(\log n)$ time to find $\overline{p_2p_3}$. We next reduce the time to $O(\log h)$.

The OvL algorithm uses the binary search strategy. In each iteration, it picks two candidate points $p'_2\in A_2$ and $p'_3\in \ch_{34}$ (initially, $p'_2$ is the middle point of $A_2$ and $p_3'$ is the middle point of $\ch_{34}$), and in $O(1)$ time, the algorithm can determine at least one of the three cases: (1) $p_2=p'_2$ and $p_3=p_3'$; (2) $p_2$ is to the left or right of $p_2'$; (3) $p_3$ is to the left or right of $p_3'$. In the first case, the algorithm stops. In the second case, half of the remaining portion of $A_2$ is pruned and $p_2'$ is reset to the middle point of the remaining portion of $A_2$ (but $p_3'$ does not change). The third case is processed analogously.

When we apply the OvL algorithm, the way we set $p_3'$ is the same as described above. However, for $p_2'$, we set it by following the exponential search strategy using $T_2$ (or the standard finger search with a finger at the leftmost leaf of $T_2$~\cite{Brodal04,GuibasMPR77,ref:Tsakalidis85}). Specifically, we first reset $p_2'$ to the leftmost node of $T_2$ and then continue the search on its parent and so on, until the first time we find $p_2$ is left of $p_2'$, in which case we search downwards on $T_2$. This will eventually find $\overline{p_2p_3}$.

We claim that the runtime is $O(\log h)$.
Indeed, since the height of $T_{34}$ is $O(\log h)$, the time the algorithm spent on resetting $p_3'$ in the entire algorithm is $O(\log h)$. To analyze the time the algorithm spent on resetting $p_2'$, following the standard finger search (with a finger at the leftmost leaf of the tree) or exponential search analysis, the number of times the algorithm resets $p_2'$ is $O(\log m)$, where $m$ is the number of points of $T_2$ to the left of $p_2$. Observe that all points of $T_2$ to the left of $p_2$ are on $\ch(P)$; thus, we have $m\leq h$. As such, the runtime of the algorithm for computing the tangent $\overline{p_2p_3}$ is $O(\log h)$.

After $\overline{p_2p_3}$ is computed, we split $T_2$ at $p_2$ and obtain a tree $T_2'$ of height $O(\log h)$ representing all points of $T_2$ left of $p_2$. Then, we merge $T_2'$ with $T_1$ and $T_{34}$ to finally obtain a tree of height $O(\log h)$ representing $\ch(P)$ in $O(\log h)$ time. The total time of the algorithm is $O(\log h)$.

As before, we use a stack to record changes due to the above algorithm for constructing the tree representing $\ch(P)$. Once we finish the queries using the tree, we roll back the changes we have made by popping the stack. This also takes $O(\log h)$ time.

\subsubsection{\HR}

To output the convex hull $\ch(P)$, we first perform \HTR\ to obtain a tree of height $O(\log h)$ that stores $\ch(P)$ in $O(\log h)$ time. Then, using this tree, we can output~$\ch(P)$ in additional $O(h)$ time. As such, the total time for reporting $\ch(P)$ is $O(h)$.

\section{The simple path problem}
\label{sec:path}
In this section, we consider the dynamic convex hull problem for a simple path.
Let $\pi$ be a simple path of $n$ vertices in the plane. Unless otherwise stated, a ``point'' of $\pi$ always refers to a vertex of it (this is for convenience also for being consistent with the notion in Section~\ref{sec:monotone}). For ease of discussion, we assume that no three vertices of $\pi$ are colinear.

For any subpath $\pi'$ of $\pi$, let $|\pi'|$ denote the number of vertices of $\pi$, and $\ch(\pi')$ the convex hull of $\pi'$, which is also the convex hull of all vertices of $\pi'$.

We designate the two ends of $\pi$ as the {\em front end} and the {\em rear end}, respectively. 
We consider the following operations on $\pi$: \IF, \DF, \IRR, \DRR, \SQ, and \HR, as defined in Section~\ref{sec:intro}. 
The following theorem summarizes the main result of this section.

\begin{theorem}
    \label{thm:dpch}
    Let $\pi \subset \mathbb{R}^2$ be an initially empty simple path, with $n = |\pi|$ and $h = |\ch(\pi)|$.
    There exists a ``Deque Path Convex Hull'' data structure $PH(\pi)$ of $O(n)$ space that supports the following operations:
    \begin{enumerate}
        \item \IF: $O(1)$ time.
        \item \DF: $O(1)$ time.
        \item \IRR: $O(1)$ time.
        \item \DRR: $O(1)$ time.
        \item \SQ: $O(\log n)$ time.
        \item \HR: $O(h+\log n)$ time.
    \end{enumerate}
\end{theorem}

\subparagraph{Remark.} We will show in Section~\ref{sec:lower-bounds} that all these bounds are optimal even for the ``one-sided'' case. In particular, it is not possible to reduce the time of \HTR\ to $O(\log h)$ or reduce the time of \HR\ to $O(h)$. This is why we do not consider the one-sided simple path problem separately. For answering standard queries, our algorithm first constructs four BSTs representing convex hulls of four (consecutive) subpaths of $\pi$ whose union is $\pi$ and then uses these trees to answer queries. 
The heights of the two trees for the two middle subpaths are $O(\log n)$ while the heights of the other two are $O(\log\log n)$.
As such, all decomposable queries can be answered in $O(\log n)$ time. We show that certain non-decomposable queries can also be answered in $O(\log n)$ time, 
such as the bridge queries. 
We could further enhance the data structure so that 
a BST of height $O(\log h)$ that represents $\ch(\pi)$ can be obtained in $O(\log n\log\log n)$ time; 
this essentially performs the \HTR\ operation (as defined in Section~\ref{sec:2sidemontone}) in $O(\log n\log\log n)$ time.
\medskip

In what follows, we prove Theorem~\ref{thm:dpch}. One crucial property we rely on is that the convex hulls of two subpaths of a simple path intersect at most twice and thus have at most two common tangents as observed by Chazelle and Guibas~\cite{ref:ChazelleFr862}. Let $\pi_1$ and $\pi_2$ be two consecutive subpaths of $\pi$. Suppose we have two BSTs representing $\ch(\pi_1)$ and $\ch(\pi_2)$, respectively. Compared to the monotone path problem, one difficulty here (we refer to it as the ``path-challenge'') is that we do not have an $O(\log n)$ time algorithm to find the common tangents between $\ch(\pi_1)$ and $\ch(\pi_2)$ and thus merge the two hulls. The best algorithm we have takes $O(\log^2 n)$ time by a nested binary search, assuming that we have two ``helper points'': a point on each convex hull that is outside the other convex hull~\cite{guibas1991compact}.

It is tempting to apply the deque convex hull idea of Theorem~\ref{thm:deque-convex-hull} (i.e., instead of considering the points in left-to-right order, we consider the points in the ``path order'' along $\pi$). We could get the same result as in Theorem~\ref{thm:deque-convex-hull} except that the \HTR\ operation now takes $O(\log^2 n)$ time and \HR\ takes $O(h+\log^2 n)$ time due to the path-challenge. As such, our main effort below is to achieve $O(\log n)$ time for \SQ\ and $O(h+\log n)$ time for \HR.

Before presenting our data structure, we introduce in Section~\ref{sec:pathpre} several basic lemmas that will be frequently used later.

\subsection{Basic lemmas}
\label{sec:pathpre}

The following lemma was given in \cite{guibas1991compact}, and we sketch the proof here to make the paper more self-contained. Later, we will need to modify the algorithm for other purposes.

\begin{lemma} {\em (Guibas, Hershberger, and Snoeyink\cite[Lemma 5.1]{guibas1991compact})}
    \label{lem:helper-points}
    Let $\pi_1$ and $\pi_2$ be two consecutive subpaths of $\pi$. Suppose the convex hull $\ch(\pi_i)$ is stored in a BST of height $O(\log |\pi_i|)$, for $i=1,2$. We can do the following in $O(\log (|\pi_1| + |\pi_2|))$ time: Determine whether $\ch(\pi_2)$ is completely inside $\ch(\pi_1)$ and if not find a ``helper point'' $p \in \partial \ch(\pi_2)$ such that $p \in \partial \ch(\pi_1\cup \pi_2)$ and $p \notin \partial \ch(\pi_1)$.
\end{lemma}
\begin{proof}
    We sketch the proof here. See \cite[Lemma 5.1]{guibas1991compact} for the details. Since $\pi_1$ and $\pi_2$ are consecutive subpaths of $\pi$, let $q_1\in \pi_1$ and $q_2\in \pi_2$ be consecutive vertices of $\pi$.

    If $q_1$ is in the interior of $\ch(\pi_1)$, then $q_1$ is in a single ``bay'' with a hull edge $e$ such that: if $\ch(\pi_2)$ is inside $\ch(\pi_1)$, then $\pi_2$ must be inside the bay; otherwise, the path $\overline{q_1q_2}\cup \pi_2$ must cross $e$ (see the left of Figure~\ref{fig:pstar}). The edge $e$ can be computed in $O(\log |\pi_1|)$ time by binary search since indices of vertices of $\ch(\pi_1)$ form a bimodal sequence\cite{guibas1991compact}. Let $\rho$ be the normal of $e$ toward outside $\ch(\pi_1)$. We compute the most extreme vertex $p$ of $\ch(\pi_2)$ along $\rho$, which takes $O(\log |\pi_2|)$ time. Then, $\ch(\pi_2)$ is inside $\ch(\pi_1)$ if and only if $p$ is in the same side of $e$ as $\ch(\pi_1)$. If $p$ is on the opposite side of $e$ as $\ch(\pi_1)$, then $p$ is a helper point as defined in the lemma statement.

    \begin{figure}[t]
        \centering
        \includegraphics[width=1.8in]{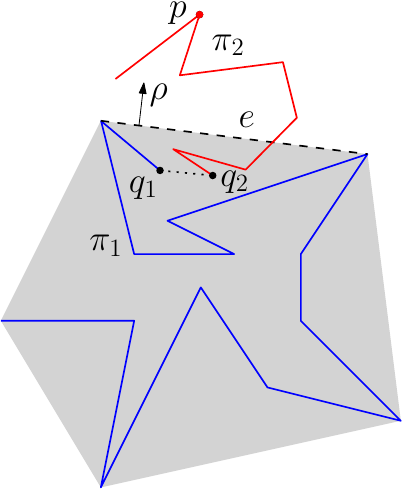}
        \hspace{0.7in}
        \includegraphics[width=1.8in]{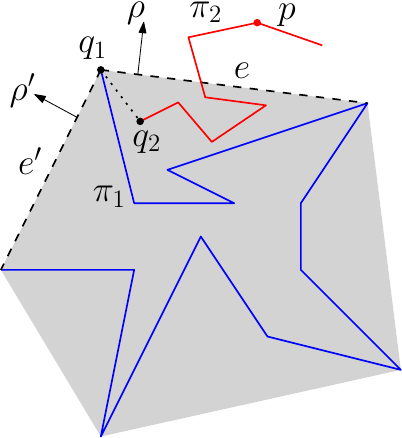}
        \caption{Illustrating the two cases in the proof of Lemma~\ref{lem:helper-points}. The grey region in each case is $\ch(\pi_1)$. $\rho$ is the normal of $e$ and $\rho'$ is the normal of $e'$.}
        \label{fig:pstar}
    \end{figure}

    If $q_1\in \partial\ch(\pi_1)$, then $q_1$ is a common vertex of two edges $e$ and $e'$ of $H_1$ (see the right of Figure~\ref{fig:pstar}), then we can determine whether $\ch(\pi_2)$ is inside $\ch(\pi_1)$ and if not find a helper point on $\partial \ch(\pi_2)$, using the normals of $e$ and $e'$, in a similar way to the above. The total time is still $O(\log (|\pi_1| + |\pi_2|))$.
\end{proof}

The following lemma is also from \cite{guibas1991compact} and uses Lemma~\ref{lem:helper-points}. We again sketch the proof.
\begin{lemma} {\em (Guibas, Hershberger, and Snoeyink\cite[Section 2]{guibas1991compact})}
    \label{lem:mergepaths}
    Let $\pi_1$ and $\pi_2$ be two consecutive subpaths of $\pi$. Suppose the convex hull $\ch(\pi_i)$ is stored in a BST of height $O(\log |\pi_i|)$, $i=1,2$. We can compute a BST of height $O(\log (|\pi_1|+|\pi_2|))$ that stores the convex hull of $\pi_1\cup \pi_2$ in $O(\log|\pi_1|\cdot\log |\pi_2|)$ time.
\end{lemma}
\begin{proof}
    We first use Lemma~\ref{lem:helper-points} to check whether one of the two convex hulls $\ch(\pi_1)$ and $\ch(\pi_2)$ contains the other. If yes, then we simply return the tree of the larger convex hull. Otherwise, we compute a helper point $p_1\in \partial \ch(\pi_1)$ and a helper point $p_2\in \partial \ch(\pi_2)$ by Lemma~\ref{lem:helper-points}. Using the two helper points, the two common tangents of $\ch(\pi_1)$ and $\ch(\pi_2)$ can be computed in $O(\log |\pi_1|\cdot \log |\pi_2|)$ time by a nested binary search~\cite{guibas1991compact}. Consequently, a tree of height $O(\log (|\pi_1|+|\pi_2|))$ that stores $\ch(\pi_1\cup \pi_2)$ can be obtained in $O(\log |\pi_1|\cdot \log |\pi_2|)$ time from the two trees for $\ch(\pi_1)$ and $\ch(\pi_2)$.
\end{proof}

The following lemma provides a basic tool for answering bridge queries.

\begin{lemma}
    \label{lem:queries}
    Let $H_1, H_2, \ldots, $ be a collection of $O(1)$ convex polygons, each represented by a BST or an array so that binary search on each convex hull can be supported in $O(\log n)$ time. Let $H$ be the convex hull of all these convex polygons. We can answer the following queries in $O(\log n)$ time each, where $n$ is the total number of vertices of all these convex polygons.
    \begin{enumerate}
        \item Bridge queries: Given a query line $\ell$, determine whether $\ell$ intersects $H$, and if yes, find the edges of $H$ that intersect $\ell$.
        \item Given a query point $p$, determine whether $p \in H$, and if yes, further determine whether $p \in \partial H$.
    \end{enumerate}
\end{lemma}
\begin{proof}
    It is not difficult to see that the second query can be reduced to the bridge queries (e.g., we can apply a bridge query on the vertical line $\ell$ through $p$, and the answer to the query on $p$ can be obtained based on the bridge query outcome). As such, in what follows, we focus on the bridge queries. Without loss of generality, we assume that $\ell$ is vertical.

    We first find the leftmost vertex $p_1$ of $H$, which can be done in $O(\log n)$ by finding the leftmost vertex of each $H_i$. Similarly, we find the rightmost vertex $p_2$ of $H$. Then, $\ell$ intersects~$H$ if and only if $\ell$ is between $p_1$ and $p_2$. Assuming that $\ell$ is between them, we next compute the edges of $H$ that intersect $\ell$. For ease of discussion, we assume that $\ell$ does not contain any vertex of any convex hulls. As such, $\ell$ intersects exactly one edge of the upper hull of $H$ and intersects exactly one edge of the lower hull. We only discuss how to find the edge on the upper hull, denoted by $e^*$, since the edge on the lower hull can be found similarly.

    Observe that there are two cases for $e^*$: (1) The two vertices of $e^*$ are from the same polygon $H_i$; (2) the two vertices of $e^*$ are from two different polygons $H_i$ and $H_j$, respectively. In the first case, $e^*$ is the edge of the upper hull of $H_i$ intersecting $\ell$. In the second case, $e^*$ is the upper tangent of $H_i$ and $H_j$, and further, $e^*$ is the edge of the upper hull of $\ch(H_i\cup H_j)$ that intersects $\ell$. Based on this observation, our algorithm works as follows.

    For any two pairs of convex polygons $H_i$ and $H_j$ whose convex hull $\ch(H_i\cup H_j)$ intersects~$\ell$ (again this can be determined by finding the leftmost and rightmost vertices of both $H_i$ and~$H_j$), we wish to compute the edge $e_{ij}$ of the upper hull of $\ch(H_i\cup H_j)$ intersecting $\ell$. Then, based on the above observation, among the edges $e_{ij}$ for all such pairs $H_i$ and $H_j$, $e^*$~is the one whose intersection with $\ell$ is the highest. We next show that $e_{ij}$ can be found in $O(\log n)$ time, meaning that $e^*$ can be found in $O(\log n)$ time too as there are $O(1)$ convex hulls.

    Let $D_i$ be the data structure (an array or a BST) storing $H_i$. Let $H_i^l$ and $H_i^r$ denote the portions of $H_i$ to the left and right of $\ell$, respectively. Regardless of whether $D_i$ is an array or a BST, binary search on each of $H_i^l$ and $H_i^r$ can be supported in $O(\log n)$ time by $D_i$. Define $D_j$, $H_j^l$, and $H_j^r$ similarly.

    Let $e^1_{ij}$ be the edge of the upper hull of $\ch(H_i^l\cup H_j^r)$ that intersects $\ell$. Let $e^2_{ij}$ be the edge of the upper hull of $\ch(H_j^l\cup H_i^r)$ that intersects $\ell$. Observe that among the above two edges, $e_{ij}$ is the one whose intersection with $\ell$ is higher. Hence, to compute $e_{ij}$, it suffices to compute $e^1_{ij}$ and $e^2_{ij}$. To compute $e^1_{ij}$, as $H_i^l$ and $H_j^r$ are separated by $\ell$, their upper common tangent $e$ can be computed in $O(\log n)$ time with the data structures $D_i$ and $D_j$, using the binary search algorithm of Overmars and van Leeuwen~\cite{Overmars}. We also compute the edge $e_i$ on the upper hull of $H_i$ intersecting $\ell$ and the edge $e_j$ of the upper hull of $H_j$ intersecting~$\ell$, which take $O(\log n)$ time. Among the three edges $e_i$, $e_j$, and $e$, the edge $e^1_{ij}$ is the one whose intersection with $\ell$ is the highest. As such, $e^1_{ij}$ can be computed in $O(\log n)$ time. The edge $e^2_{ij}$ can be computed analogously in $O(\log n)$ time. Therefore, $e_{ij}$ can be computed in $O(\log n)$ time. In summary, each bridge query can be answered in $O(\log n)$ time.
\end{proof}

\subsection{Structure of the deque path convex hull \texorpdfstring{$\boldsymbol{PH(\pi)}$}{PH(pi)}}
\label{sec:ph-structure}
For convenience, we assume that $|\pi|$ is always greater than a constant number (e.g., $|\pi|> 8$). 
We partition $\pi$ into four (consecutive) subpaths $\pi_r$, $\pi_m^r$, $\pi_m^f$, and $\pi_f$ in the order from the front to the rear of $\pi$. As such, $\pi_f$ and $\pi_r$ contain the front and rear ends, respectively. 
Further, let $\pi^+=\pi_f\cup\pi_m^f$ and $\pi^-=\pi_r\cup\pi_m^r$. Our algorithm maintains the following invariants.
\begin{enumerate}
    \item $\frac{1}{4} \leq  |\pi^+| / |\pi^-| \leq  4$.
    \item $|\pi_f| = O(\log^2 |\pi^+|)$ and $|\pi_r| = O(\log^2 |\pi^-|)$.
\end{enumerate}

Note that the invariants imply that $|\pi_f|, |\pi_r|=O(\log^2n)$ and $|\pi_m^f|, |\pi_m^r|=\Theta(n)$, where $n=|\pi|$. The first invariant resembles the partition of $P$ by a dividing line $\ell$ in our deque tree in Section~\ref{sec:dequetree}. As with the deque tree, in order to maintain the first invariant, whenever $|\pi^+| / |\pi^-|<1/2$ or $|\pi^+| / |\pi^-| > 2$, we pick a new center to partition $\pi$ and start to build a deque path convex hull data structure by incrementally inserting points around this new center toward both front and rear ends. As will be seen later, since each insertion takes $O(1)$ time, it takes $O(n)$ time to build this new data structure. Thus, we can spread the incremental work over the next $\Theta(n)$ updates so that only $O(1)$ incremental work on each update is incurred to maintain this invariant.

\begin{figure}[t]
    \centering
    \includegraphics[width=\linewidth]{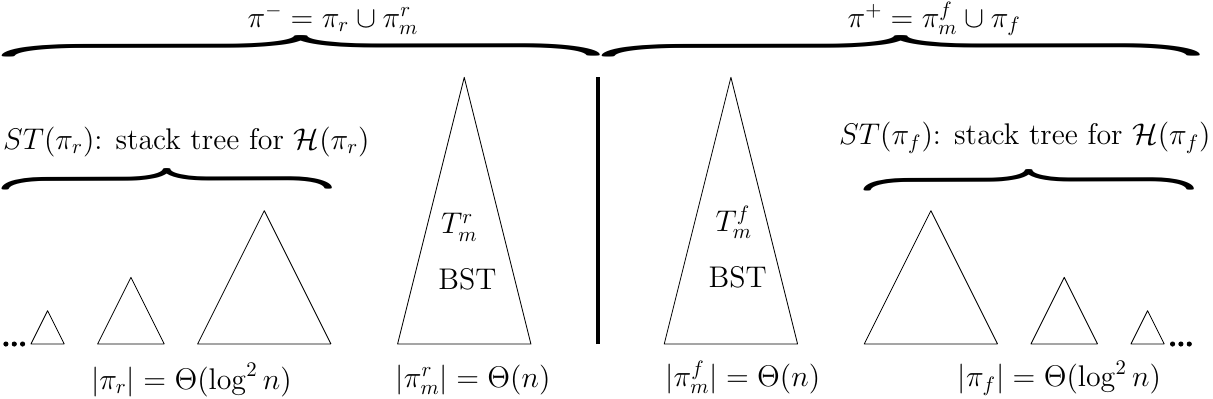}
    \caption{A schematic view of the deque path convex hull data structure $PH(\pi)$.}
    \label{fig:path-convex-hull}
\end{figure}

We use a stack tree $ST(\pi_f)$ to maintain the convex hull $\ch(\pi_f)$, with the algorithm in Lemma~\ref{lem:mergepaths} for merging two hulls of two consecutive subpaths. More specifically, we consider the vertices of $\pi_f$ following their order along the path (instead of left-to-right order as in Section~\ref{sec:stacktree}) with insertions and deletions only at the front end. Whenever we need to join two neighboring trees, we merge the two hulls of their corresponding subpaths by Lemma~\ref{lem:mergepaths}.\footnote{The analysis is slightly different since we use a merging algorithm of bigger complexity. More specifically, merging the hulls represented by the trees $T_{i}$ and $T_{i+1}$ now takes $O(\log^2|T_i|)$ time, which is $O(2^{2i+2})$. Nevertheless, we can still guarantee that the current incremental joining process for $T_{i+1}$ is completed before the next joining process starts by slightly modifying the proof of Lemma~\ref{lem:invariantstacktree}. More specifically, by making the constant $c$ slightly larger, the joining process can be completed within $2^{3i-1}$ insertions. After $2^{3i-1}$ insertions, the total number of points that can be inserted to $T_{i}$ is at most $2\cdot 2^{2^{i}}+2^{3i-1}$, which is still smaller than $2^{2^{i+1}}$ for $i\geq 1$.} Due to the second invariant, merging all trees of $ST(\pi_f)$ takes $O(\log^2\log n)$ time, after which we obtain a single tree of height $O(\log\log n)$ that represents $\ch(\pi_f)$. Similarly, we build a stack tree $ST(\pi_r)$ for $\ch(\pi_r)$ but along the opposite direction of the path. See Figure~\ref{fig:path-convex-hull} for an illustration.

Define $n^+=|\pi^+|$. 
In order to maintain the second invariant, when $\pi_f$ is too big due to insertions, we will cut a subpath of length $\Theta(\log^2 n^+)$ and concatenate it with $\pi_m^f$. When $\pi_f$ becomes too small due to deletions, we will split a portion of $\pi_m^f$ of length $\Theta(\log^2 n^+)$ and merge it with $\pi_f$; but this split is done implicitly using the rollback stack for deletions. As such, we need to build a data structure for maintaining $\pi_m^f$ so that the above concatenate operation on $\pi_m$ can be performed in $O(\log^2 n^+)$ time (this is one reason why the bound for~$\pi_f$ in the second invariant is set to $O(\log^2 n^+)$). We process $\pi^-$ in a symmetric way. 
The way we handle the interaction between $\pi_m^f$ and $\pi_f$ (as well as their counterpart for $\pi^-$) are one main difference from our approach for the two-sided monotone path problem in Section~\ref{sec:2sidemontone}; again this is due to the path-challenge.

Our data structure for $\pi_m^f$ is simply a balanced BST $T_m^f$, which stores the convex hull~$\ch(\pi_m^f)$. In particular, we will use $T_m^f$ to support the above concatenation operation (denoted by \CON) in $O(\log^2 n)$ time. For reference purpose, this is summarized in the following lemma, which is an immediate application of Lemma~\ref{lem:mergepaths}. 

\begin{lemma}
    \label{lem:pcht}
    Given a BST of height $O(\log|\tau|)$ representing a simple path $\tau$ of length $O(\log^2n)$ such that the concatenation of $\pi_f^m$ and $\tau$ is still a simple path, we can perform the following \CON\ operation in $O(\log^2n)$ time: Obtain a new tree $T_m^f$ of height $O(\log n)$ that represents the convex hull $\ch(\pi_m^f)$, where $\pi_m^f$ is the new path after concatenating with $\tau$.
\end{lemma}

Similarly, we use a balanced BST $T_m^r$ to store the convex hull $\ch(\pi_m^r)$. We have a similar lemma to the above for the \CON\ operation on $\pi_m^r$.

The four trees $ST(\pi_r)$, $T_m^r$, $T_m^f$, and $ST(\pi_f)$ constitute our deque path convex hull data structure $PH(\pi)$ for Theorem~\ref{thm:dpch}; see Figure~\ref{fig:path-convex-hull}. In the following, we discuss the operations.

\subsection{Standard queries}

For answering a decomposible query $\sigma$, we first perform a \TR\ operation on $ST(\pi_f)$ to obtain a tree $T_f$ that represents $\ch(\pi_f)$. Since $|\pi_f|=O(\log^2 n)$, this takes $O(\log^2\log n)$ time as discussed before. We do the same for $ST(\pi_r)$ to obtain a tree $T_r$ for~$\ch(\pi_r)$. 
Recall that the tree $T_m^f$ stores $\ch(\pi_m^f)$ while $T_m^r$ stores $\ch(\pi_m^r)$. 
We perform query~$\sigma$ on each of the above four trees. Based on the answers to these trees, we can obtain the answer to the query~$\sigma$ for $\ch(\pi)$ because $\sigma$ is a decomposable query. Since the heights of~$T_f$ and $T_r$ are both $O(\log\log n)$, and the heights of $T_m^f$ and $T_m^r$ are $O(\log n)$, the total query time is $O(\log n)$.

If $\sigma$ is a bridge query, we apply Lemma~\ref{lem:queries} on the above four trees. The query time is $O(\log n)$.

\subsection{Insertions and deletions}
\label{sec:phupdate}

We now discuss the updates. 
\IF\ and \DF\ will be handled by the data structure for $\pi^+$, i.e., $T_m^f$ and $ST(\pi_f)$, while  \IRR\ and \DRR\ will be handled by the data structure for $\pi^-$.

\subparagraph{\IF.}
Suppose we insert a point $p$ to the front end of $\pi$. We first perform the insertion using the stack tree $ST(\pi_f)$. To maintain the second invariant, we must handle the interaction between the largest tree $T_k$ of $ST(\pi_f)$ and the tree $T_m^f$. Recall that $n^+=|\pi^+|$ and $n^+=\Theta(n)$. 

According to the second invariant and the definition of the stack tree $ST(\pi_f)$, we have $|T_k|=O(\log^2 n^+)$, and we can assume a constant $c$ such that the total size of all trees of $ST(\pi_f)$ smaller than $T_k$ is at most $c\cdot \log^2 n^+$. We set the size of $T_k$ to be $(c+1)\cdot \log^2 n^+$. During the algorithm, whenever $|T_k|> (c+1)\cdot \log^2 n^+$ and there is no incremental process of joining $T_{k-1}$ with $T_k$, we let $T_k'=T_k$ and let $T_k=\emptyset$, and then start to perform an incremental \CON\ operation to concatenate $T_k'$ with $T_m^f$. The operation takes $O(\log^2n^+)$ time by Lemma~\ref{lem:pcht}. We choose a sufficiently large constant $c_1$ so that each \CON\ operation can be finished within $c_1\cdot \log^2n^+$ steps. For each \IF\ in future, we run $c_1$ steps of this \CON\ algorithm. As such, within the next $\log^2 n^+$ \IF\ operations in future, the \CON\ operation will be completed. If there is an incremental \CON\ operation (that is not completed), then we say that $T_m^f$ is {\em dirty}; otherwise, it is {\em clean}.

If $T_m^f$ is dirty, an issue arises during a \SQ\ operation. Recall that during a \SQ\ operation, we need to perform queries on $\ch(\pi_m^f)$ by using the tree~$T_m^f$. However, if $T_m^f$ is dirty, we do not have complete information for $T_m^f$. To address this issue, we resort to persistent data structures~\cite{ref:DriscollMa89,ref:SarnakPl86}. Specifically, we use a persistent tree for $T_m^f$ so that if there is an incremental \CON\ operation, the old version of $T_m^f$ can still be accessed (we call it the ``clean version''); as such, a partially persistent tree suffices for our purpose~\cite{ref:DriscollMa89,ref:SarnakPl86}. After the \CON\ is completed, we designate the new version of $T_m^f$ as clean and the old version as dirty; in this way, at any time, there is only one clean version we can refer to. During a \SQ\ operation, we can perform queries on the clean version of $T_m^f$. Similarly, during the query, if there is an incremental \CON\ process, $T_k'$ is also  dirty, and we need to access its clean version (i.e., the version right before $T_k'$ started the \CON\ operation). To solve this problem, before we start \CON, we make another copy of $T_k'$, denoted by $T_k''$. After \CON\ is completed, we make~$T_k''$ refer to null. The above strategy causes additional $O(\log^2 n^+)$ time, i.e., update the persistent tree $T_m^f$ and make a copy $T_k''$. To accommodate this additional cost, we make the constant $c_1$ large enough so that all these procedures can be completed within the next $\log^2 n^+$ \IF\ operations.

Recall that once we are about to start a \CON\ operation for $T_k'$, $T_k$ becomes empty. We argue that \CON\ will be completed before another \CON\ operation starts. To this end, it suffices to argue that the current \CON\ will be completed before $|T_k|$ becomes larger than $(c+1)\cdot \log^2 n^+$ again. Indeed, we know that the current \CON\ will be finished within the next $\log^2 n^+$ \IF\ operations. On the other hand, recall that the number of points in all trees of $ST(\pi_f)$ smaller than~$T_k$ is at most $c\cdot \log^2 n^+$. Since all points of $T_k$ come from those smaller trees plus newly inserted points from the \IF\ operations, within the next $\log^2 n^+$ \IF\ operations, $|T_k|$ cannot be larger than $(c+1)\cdot \log^2 n$. As such, there cannot be two concurrent \CON\ operations from $T_k$ to $T_m^f$.

\subparagraph{\DF.}
To perform a \DF\ operation, i.e., delete the front vertex $p$ of~$\pi$, we keep a stack of changes to our data structure $PH(\pi)$ due to the \IF\ operations. When deleting $p$, $p$ must be the most recently inserted point at the front end, and thus, the changes to the data structure due to the insertion of $p$ must be at the top of the stack. We roll back these changes by popping the stack. 

\subparagraph{\IRR\ and \DRR.}
Handling updates at the rear end is the same, but using $T_m^r$ and $ST(\pi_r)$ instead. We omit the details.

\subsection{Reporting the convex hull \texorpdfstring{$\boldsymbol{\ch(\pi)}$}{H(pi)}}

We show that the convex hull $\ch(\pi)$ can be reported in $O(h+\log n)$ time.

As in the algorithm for \SQ, we first obtain in $O(\log n)$ time the four trees $T_f$, $T_r$, $T_m^f$, and $T_m^r$ representing $\ch(\pi_f)$, $\ch(\pi_r)$, $\ch(\pi_m^f)$, and $\ch(\pi_m^r)$, respectively. Then, we can merge these four convex hulls using Lemma~\ref{lem:mergepaths} in $O(\log^2 n)$ time and compute a BST~$T(\pi)$ representing $\ch(\pi)$. Finally, we can output $\ch(\pi)$ by traversing $T(\pi)$ in additional $O(h)$ time. As such, in total $O(h+\log^2 n)$ time, $\ch(\pi)$ can be reported. In what follows, we reduce the time to $O(h+\log n)$. To this end, we first enhance our data structure $PH(\pi)$ by having it maintain the common tangents of the convex hulls of the two middle subpaths $\pi_m^f$ and $\pi_m^r$ during updates.

\subparagraph{Remark.} As will be discussed at the end of this section, it is possible to achieve $O(h+\log n)$ time for \HR\ without enhancing the data structure. Nevertheless, we choose to present the enhanced data structure for two reasons: (1) Enhancing the data structure will make the \HR\ algorithm much simpler; (2) the enhanced data structure helps us to obtain in $O(\log n\log\log n)$ time a tree of height $O(\log n)$ to represent $\ch(\pi)$, improving the aforementioned $O(\log^2 n)$ time algorithm.

\subsubsection{Enhancing the data structure \texorpdfstring{$\boldsymbol{PH(\pi)}$}{PH(pi)}}

We now enhance our data structure $PH(\pi)$ described in Section~\ref{sec:ph-structure}. We have our enhanced~$PH(\pi)$ maintain the common tangents between the convex hulls of the two middle subpaths of $\pi$, i.e., $\ch(\pi_m^f)$ and $\ch(\pi_m^r)$. Before describing how to maintain these common tangents, we first explain why they are useful. Let $\pi_m$ denote the concatenation of $\pi_m^f$ and~$\pi_m^r$.

Suppose the common tangents of $\ch(\pi_m^f)$ and $\ch(\pi_m^r)$ are available to us. Then, we can compute in $O(\log n)$ time a tree $T_m$ representing $\ch(\pi_m)$ by merging $\ch(\pi_m^f)$ and $\ch(\pi_m^r)$ using their trees $T_m^f$ and $T_m^r$. Consequently, we can do the following. (1) With the three trees $T_f$, $T_m$, and $T_r$, representing $\ch(\pi_f)$, $\ch(\pi_m)$, and $\ch(\pi_r)$, respectively, we can report~$\ch(\pi)$ in $O(h+\log n)$ time and the algorithm will be given later. (2) By further merging $\ch(\pi_f)$, $\ch(\pi_m)$, and $\ch(\pi_r)$ using their trees, we can obtain a single tree of height $O(\log n)$ representing $\ch(\pi)$, which takes $O(\log n\log\log n)$ time by Lemma~\ref{lem:mergepaths} since the heights of both $T_f$ and $T_r$ are $O(\log\log n)$ and the height of $T_m$ is $O(\log n)$.

We proceed to describe how to maintain the common tangents of $\ch(\pi_m^f)$ and $\ch(\pi_m^r)$. To do so, we modify our algorithms for the update operations so that whenever $\pi_m^f$ or $\pi_m^r$ is changed (either due to the \CON\ operation or the rollback of the operation, referred to as an {\em inverse} \CON), we recompute the common tangents. By Lemma~\ref{lem:mergepaths}, their common tangents can be computed in $O(\log^2 n)$ time. The main idea is to incorporate the tangent-computing algorithm into the \CON\ operation. Once an incremental \CON\ operation is finished, say, on $\pi_m^f$, we immediately start computing the new common tangents. We consider this tangent-computing procedure part of the \CON\ operation. As each \CON\ operation takes $O(\log^2 n)$ time, incorporating the $O(\log^2 n)$ time tangent-computing algorithm as above does not change the time complexity of \CON\ asymptotically. However, we do have an issue with this idea. Recall that during \DF\ operations, we roll back the changes incurred by \IF\ operations, and thus, each inverse \CON\ is done by the rollback. After an inverse \CON\ is completed, $\pi_m^f$ loses a subpath of length $\Theta(\log^2 n^+)$, and we need to recompute the common tangents. However, in this case, rollback cannot recompute the common tangents; thus, we explicitly run the tangent-computing algorithm during rollback. As such, we must find an effective way to incorporate this additional step into the rollback process so that each \DF\ operation still takes $O(1)$ worst-case time. For reference purpose, we call the above the {\em rollback issue}.

We now elaborate on how to modify our update algorithms. We only discuss \IF\ and \DF\ since the other two updates are similar.

\subparagraph{\IF.}
We follow notation in Section~\ref{sec:phupdate}, e.g., $T_m^f$, $ST(\pi_f)$, $n^+$, $T_k$, $T_k'$, $T_{k-1}$, $c$, $c_1$, etc. During the algorithm, as before, whenever $|T_k|> (c+1)\cdot \log^2 n^+$ and there is no incremental process of joining $T_{k-1}$ with $T_k$, we let $T_k'=T_k$, $T_k=\emptyset$, and start an incremental \CON\ to concatenate $T_k'$ with $T_m^f$. Recall that the incremental \CON\ in Section~\ref{sec:phupdate} will be completed within the next $\log^2 n^+$ \IF\ operations. Here, to complete \CON, we do the following three steps in the next $\log^2 n^+$ \IF\ operations (again, these are ``net'' \IF\ operations, i.e., the number of \IF\ operations minus the number of \DF\ operations happened in the future): (1) For each of the next $\frac{1}{3}\cdot \log^2 n^+$ \IF\ operations, we push a ``token'' into our rollback stack, where a token could be a special symbol. Intuitively, a token represents a time credit we can use during the rollback to perform an incremental tangent-computing procedure. This is our mechanism to address the aforementioned rollback issue. (2) In each of the subsequent $\frac{1}{3}\cdot\log^2 n^+$ \IF\ operations, we perform the next $3c_1$ steps of the actual \CON\ operation by Lemma~\ref{lem:pcht} (including the extra step of making a copy $T_k''$ from~$T_k'$ first as discussed in Section~\ref{sec:phupdate}). (3) Start an incremental tangent-computing algorithm to find the common tangents between this new $\ch(\pi_m^f)$ (using the new $T_m^f$ after the current \CON\ is completed in the above step (2)) and $\ch(\pi_m^r)$ (using the clean version of its tree $T_m^r$), and in each of the subsequent $\frac{1}{3}\cdot\log^2 n^+$ \IF\ operations, perform the next $c_2$ steps of the tangent-computing algorithm, for a sufficiently large constant $c_2$ so that the algorithm will be finished within $\frac{c_2}{3}\cdot\log^2 n^+$ steps (such a constant $c_2$ exists because the tangent-computing algorithm runs in $O(\log^2 n)$ time and $n^+=\Theta(n)$). Once all three steps are completed, we consider the current incremental \CON\ completed and set $T_k''$ to null (meaning that the subpath stored in $T_k'$ now officially becomes part of $\pi_m^f$). However, we mark the current version of $T_m^f$ clean after step (2) is completed. Since it still takes $\log^2 n^+$ \IF\ operations to complete the new incremental \CON\ operation, the prior argument in Section~\ref{sec:phupdate} that the current \CON\ will be completed before the next \CON\ starts still applies here.

\subparagraph{\DF.}
We perform the rollback process for each \DF\ operation as before but with the following changes. Whenever an ``inverse'' \CON\ is completed (i.e., the above step (2)), then according to our new incremental \CON, the rollback stack contains from the top of the stack $\frac{1}{3}\cdot\log^2 n^+$ tokens for the next $\frac{1}{3}\cdot\log^2 n^+$ to-be-deleted points (and the history of changes of the next inverse \CON\ is stored below all these tokens). As such, we mark the current $T_m^f$ clean and start a tangent-computing algorithm to find the common tangents between this new $\ch(\pi_m^f)$ and the current clean version of $T_m^r$ for $\ch(\pi_m^r)$, and for each of the subsequent $\frac{1}{3}\cdot\log^2 n^+$ tokens popped out of the rollback stack, we perform the next $c_2$ steps of the tangent-computing algorithm ($c_2$ was already defined above in such a way that the algorithm will be finished within $\frac{c_2}{3}\cdot\log^2 n^+$ steps). We consider the inverse \CON\ operation completed once the tangent-computing algorithm is finished.

\subparagraph{Obtaining the common tangents.}
We argue that the common tangents of $\ch(\pi_m^f)$ and $\ch(\pi_m^r)$ can be accessed at any time. Consider the time $t$ after an update operation. Each of $T_m^f$ and $T_m^r$ has a clean version at $t$. Suppose the current clean version $T_m^f$ (resp., $T_m^r$) became clean at time $t_1$ (resp., $t_2$). Without loss of generality, we assume $t_1\geq t_2$, i.e., $T_m^f$ became clean later than $T_m^r$ did. Hence, once $T_m^f$ became clean at $t_1$, a tangent-computing algorithm started to compute common tangents between the clean versions of $T_m^f$ and $T_m^r$.

At the time $t$, each of $T_m^f$ and $T_m^r$ may or may not have an incomplete incremental \CON\ operation or an incomplete inverse \CON\ operation. If $T_m^f$ has an incomplete operation, we let $\hat{T}_m^f$ denote its version before the operation started. Define $\hat{T}_m^r$ similarly. We use $T_m^f$ and $T_m^r$ to refer to their current clean versions at $t$. Depending on whether $T_m^f$ and/or $T_m^r$ has an incomplete operation, there are four cases. 

\begin{enumerate}
    \item
          If neither $T_m^f$ nor $T_m^r$ has an incomplete operation at $t$, then the tangent-computing algorithm for $T_m^f$ and $T_m^r$ has been finished (otherwise the operation would be incomplete at $t$). Thus, their common tangents are available at $t$.

    \item If $T_m^f$ has an incomplete operation while $T_m^r$ does not, then the common tangents between $T_m^r$ and $\hat{T}_m^f$ must be available. Indeed, as above, the one of $T_m^r$ and $\hat{T}_m^f$ that became clean later must start a tangent-computing algorithm to compute their common tangents and the algorithm must have finished before $t$ since otherwise at least one of them must have an incomplete operation at $t$, which is not true.

          Since $T_m^f$ still has an incomplete operation, $\hat{T}_m^f$ actually represents the current subpath $\pi_m^f$. More specifically, if the incomplete operation is \CON, then it is concatenating a new subpath $\tau$ to $\pi_m^f$. As the operation is incomplete, we still consider $\tau$ part of $\pi_f$ (instead of part of $\pi_m^f$), and $\tau$ is stored at $T_k''$. If the incomplete operation is an inverse \CON, then it is removing a subpath $\tau$ from $\pi_m^f$. As the operation is incomplete, we still consider $\tau$ as part of $\pi_f$, which is stored in $\hat{T}_m^f$. Therefore, the common tangents of $T_m^r$ and $\hat{T}_m^f$ are what we need.

    \item If $T_m^r$ has an incomplete operation while $T_m^f$ does not, this case is symmetric to the above second case and can be treated likewise.

    \item If both $T_m^f$ and $T_m^r$ have incomplete operations, then by a similar argument to the above, the common tangents between $\hat{T}_m^r$ and $\hat{T}_m^f$ must be available, and their common tangents are what we need.
\end{enumerate}

\subsubsection{Algorithm for \HR}
Suppose we have the three trees $T_f$, $T_m$, and $T_r$, representing $\ch(\pi_f)$, $\ch(\pi_m)$, and $\ch(\pi_r)$, as discussed above. We now describe our algorithm for reporting $\ch(\pi)$ using the three trees. Depending on whether one of the three convex hulls $\ch(\pi_f)$, $\ch(\pi_r)$, and $\ch(\pi_m)$ contains another or both of the other two, there are a number of cases.

\begin{enumerate}

    \item If $\ch(\pi_m)$ contains both $\ch(\pi_f)$ and $\ch(\pi_r)$, then $\ch(\pi_m)$ is $\ch(\pi)$ and thus we can simply report $\ch(\pi_m)$, which can be done in $O(h)$ time using the tree $T_m$. We can apply Lemma~\ref{lem:helper-points} to determine this case. More specifically, since $\pi_m$ and $\pi_f$ are consecutive subpaths of~$\pi$, applying Lemma~\ref{lem:helper-points} with $T_m$ and $T_f$ can determine whether $\ch(\pi_m)$ contains $\ch(\pi_f)$ in $O(\log n)$ time. Similarly, whether $\ch(\pi_m)$ contains $\ch(\pi_r)$ can be determined in $O(\log n)$ time. As such, in this case, we can report $\ch(\pi)$ in $O(h+\log n)$ time.

    \item If $\ch(\pi_m)$ contains $\ch(\pi_r)$ but not $\ch(\pi_f)$, and $\ch(\pi_f)$ does not contain $\ch(\pi_m)$, then we will present an algorithm later to report $\ch(\pi)$ in $O(h+\log n)$ time. As in the first case, we can determine whether this case happens in $O(\log n)$ time using Lemma~\ref{lem:helper-points}.

    \item If $\ch(\pi_m)$ contains $\ch(\pi_f)$ but not $\ch(\pi_r)$, and $\ch(\pi_r)$ does not contain $\ch(\pi_m)$, then this is a case symmetric to the above second case and thus can be treated likewise.

    \item If $\ch(\pi_r)$ contains both $\ch(\pi_m)$ and $\ch(\pi_f)$, then $\ch(\pi_r)$ is $\ch(\pi)$ and thus we can simply report $\ch(\pi_r)$, which can be done in $O(h)$ time using the tree $T_r$. We can determine whether this case happens in $O(\log n)$ time using Lemma~\ref{lem:helper-points}. Indeed, we first apply Lemma~\ref{lem:helper-points} on $\ch(\pi_r)$ and $\ch(\pi_m)$ to determine whether $\ch(\pi_r)$ contains $\ch(\pi_m)$ in $O(\log n)$. If yes, we must determine whether $\ch(\pi_r)$ contains $\ch(\pi_f)$. Note that this time we cannot apply Lemma~\ref{lem:helper-points} directly as $\pi_r$ and $\pi_f$ are not consecutive. To overcome this issue, since we already know that $\ch(\pi_r)$ contains $\ch(\pi_m)$, it suffices to know whether $\ch(\pi_r)$ contains $\ch(\pi_m\cup \pi_f)$. For this, we can slightly modify the algorithm in Lemma~\ref{lem:helper-points} since $\pi_r$ and $\pi_m\cup\pi_f$ are two consecutive subpaths of $\pi$. Specifically, the algorithm needs to solve a subproblem that is to compute the most extreme point $p$ of $\ch(\pi_m\cup \pi_f)$ along a direction $\rho$ that is the normal of an edge of $\ch(\pi_r)$. We cannot solve this subproblem by applying the algorithm in the proof of Lemma~\ref{lem:helper-points} since we do not have a tree to represent $\ch(\pi_m\cup \pi_f)$. Instead, we can compute $p$ by first computing the most extreme point of $\ch(\pi_m)$ along $\rho$ in $O(\log n)$ time using the tree $T_m$ and computing the most extreme point of $\ch(\pi_r)$ along $\rho$ in $O(\log\log n)$ time using the tree $T_r$, and then return the more extreme point along $\rho$ among the two points. As such, we can determine whether $\ch(\pi_r)$ contains $\ch(\pi_m\cup \pi_f)$ in $O(\log n)$ time.

    \item If $\ch(\pi_f)$ contains both $\ch(\pi_m)$ and $\ch(\pi_r)$, this is a case symmetric to the above fourth case and thus can be treated likewise.

    \item The remaining case: $\ch(\pi_m)$ contains neither $\ch(\pi_f)$ nor $\ch(\pi_r)$, and neither $\ch(\pi_f)$ nor $\ch(\pi_r)$ contains $\ch(\pi_m)$. We will present an $O(h+\log n)$ time algorithm. This is the most general case, and the algorithm is also the most complicated. Note that it is possible that one of $\ch(\pi_f)$ and $\ch(\pi_r)$ contains the other; however, we cannot apply Lemma~\ref{lem:helper-points} to determine that since $\pi_f$ and $\pi_r$ are not consecutive subpaths of $\pi$.
\end{enumerate}

It remains to present algorithms for the above Case 2 and Case 6.

\subparagraph{Algorithm for Case 2.}
In Case 2, $\ch(\pi_m)$ contains $\ch(\pi_r)$ but not $\ch(\pi_f)$, and $\ch(\pi_f)$ does not contain $\ch(\pi_m)$. One may consider the algorithm for this case a ``warm-up'' for the more complicated algorithm of Case 6 given later. In this case, $\ch(\pi)$ is the convex hull of $\ch(\pi_f)$ and $\ch(\pi_m)$, and thus it suffices to merge $\ch(\pi_f)$ and $\ch(\pi_m)$ by computing their common tangents. Note that computing the common tangents can be easily done in $O(\log n\log\log n)$ time by a nested binary search \cite{guibas1991compact} since the height of $T_m$ is $O(\log n)$ and the height of $T_f$ is $O(\log\log n)$. After the common tangents are computed, $\ch(\pi)$ can be output in additional $O(h)$ time. As such, the total time is $O(h+\log n\log\log n)$. In what follows, we present an improved algorithm of $O(h+\log n)$ time.

\begin{figure}[t]
    \centering
    \includegraphics[height=2.0in]{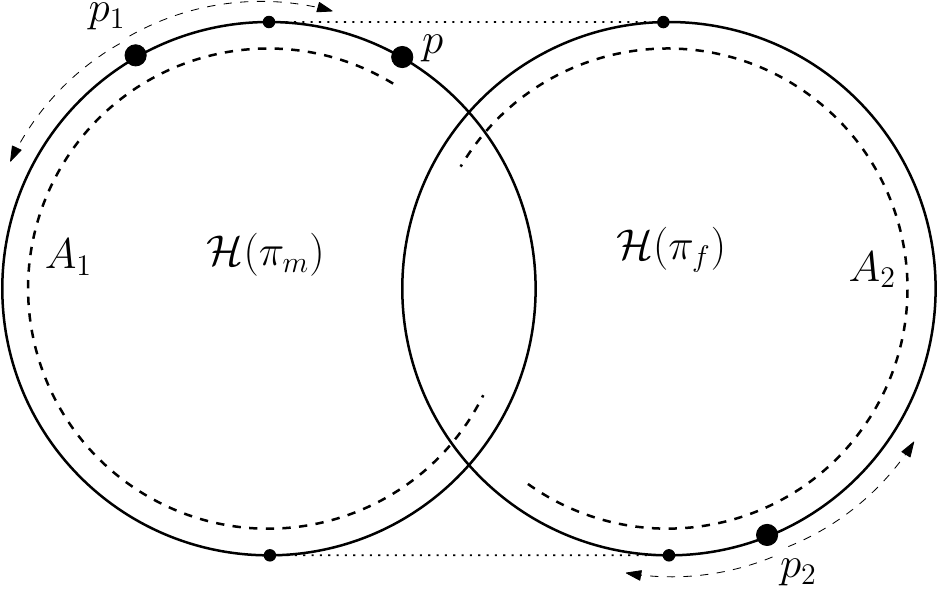}
    \caption{Illustrating the walking procedure for Case 2. The dashed arcs illustrate $A_1$ and $A_2$ on $\partial \ch(\pi_m)$ and $\partial\ch(\pi_f)$, respectively}
    \label{fig:report2h}
\end{figure}

First, we apply Lemma~\ref{lem:helper-points} to compute a helper point $p_1\in \partial \ch(\pi_m)$ such that $p_1\in \partial \ch(\pi)$ and $p_1\not\in \ch(\pi_f)$, and also compute a helper point $p_2\in \partial\ch(\pi_f)$ such that $p_2\in \partial \ch(\pi)$ and $p_2\not\in \ch(\pi_m)$. This takes $O(\log n)$ time using the trees $T_m$ and $T_f$ by Lemma~\ref{lem:helper-points}.

Next, starting from $p_1$, we perform the following {\em walking procedure}. We walk clockwise on $\partial \ch(\pi_m)$ with a ``step size'' $\log n$ (see Figure~\ref{fig:report2h}). The idea is to identify a portion of $\partial \ch(\pi_f)$ of size at most $\log n$ that contains a common tangent point. To this end, using the tree $T_m$, we find in $O(\log n)$ time a point $p$ on $\partial \ch(\pi_m)$ such that the portion $\partial \ch(\pi_m)[p_1,p]$ from $p_1$ clockwise to $p$ along $\partial \ch(\pi_m)$ contains $\log n$ vertices. Then, we determine whether $p\in \partial \ch(\pi)$. Since $\ch(\pi)$ is the convex hull of $\ch(\pi_m)$ and $\ch(\pi_f)$, we can apply Lemma~\ref{lem:queries} to determine whether $p\in \partial \ch(\pi)$ in $O(\log n)$ time.

\begin{itemize}
    \item
          If $p \not\in \partial\ch(\pi)$, then $\partial \ch(\pi_m)[p_1,p]$ must contain a tangent point. In this case, starting from $p_1$, we repeat the same walking procedure counterclockwise until we find another boundary portion of $\ch(\pi)$ that contains a tangent point. Let $A_1$ denote the set of vertices that have been traversed during the above clockwise and counterclockwise walks (see Figure~\ref{fig:report2h}; note that all vertices in each step, e.g., all vertices in $\partial \ch(\pi_m)[p_1,p]$, are included in $A_1$). Observe that $A_1$ has the following {\em key property}: Vertices of $\ch(\pi)$ on $\partial \ch(\pi_m)$ are all in $A_1$ and $A_1$ contains at most $\log n$ vertices that not on $\partial\ch(\pi)$. As such, we have $|A_1|=O(h+\log n)$.

    \item
          If $p\in\partial \ch(\pi)$, then we keep walking clockwise as above until we make a step where $p\not\in \partial \ch(\pi)$ or we pass over the starting point $p_1$ again. In the first case, we find a portion of $\log n$ vertices on $\partial \ch(\pi)$ that contain a tangent point. Then, we repeat the same walk procedure counterclockwise around $\ch(\pi)$. We can still obtain a set $A_1$ of vertices of $\partial \ch(\pi)$ with the above key property. In the second case, we realize that the last walking step contains two common tangent points. In this case, the observation is that all but at most $\log n$ vertices of $\ch(\pi_m)$ are also vertices of $\ch(\pi)$; we define $A_1$ as the set of all vertices of $\ch(\pi_m)$, and thus the key property still holds on $A_1$.
\end{itemize}

In either case, we have found a subset $A_1$ of vertices of $\ch(\pi_m)$ with the above key property. Since $|A_1|=O(h+\log n)$, the walking procedure computes $A_1$ in $O(h+\log n)$ time because each step takes $O(\log n)$ time and each step either traverses all points of $A_1$ or traverses $\log n$ points of $A_1$ (therefore the number of steps is $O(h/\log n)$). Also, since all vertices of $A_1$ are on $\partial\ch(\pi_m)$, we can sort them from left to right in linear time. Using the same strategy, starting from $p_2$ on $\ch(\pi_f)$, we can also find a subset $A_2$ of sorted vertices of $\ch(\pi_f)$ with a similar key property: Vertices of $\ch(\pi)$ on $\partial \ch(\pi_f)$ are all in $A_2$ and $A_2$ contains at most $\log n$ vertices that not on $\partial\ch(\pi)$ (see Figure~\ref{fig:report2h}). As such, the convex hull of $A_1$ and $A_2$ is exactly $\ch(\pi)$. As each of $A_1$ and $A_2$ is already sorted, their convex hull can be computed in $O(|A_1|+|A_2|)$ time, which is $O(h+\log n)$ since both $|A_1|$ and $|A_2|$ are bounded by $O(h+\log n)$.

In summary, we can report $\ch(\pi)$ in $O(h+\log n)$ time in Case 2.

\subparagraph{Algorithm for Case 6.}
We now describe the algorithm for Case 6, in which $\ch(\pi_m)$ contains neither $\ch(\pi_f)$ nor $\ch(\pi_r)$, and neither $\ch(\pi_f)$ nor $\ch(\pi_r)$ contains $\ch(\pi_m)$. The algorithm, which extends our above algorithm for Case 2,  becomes more involved.

\begin{figure}[t]
    \centering
    \includegraphics[height=2.0in]{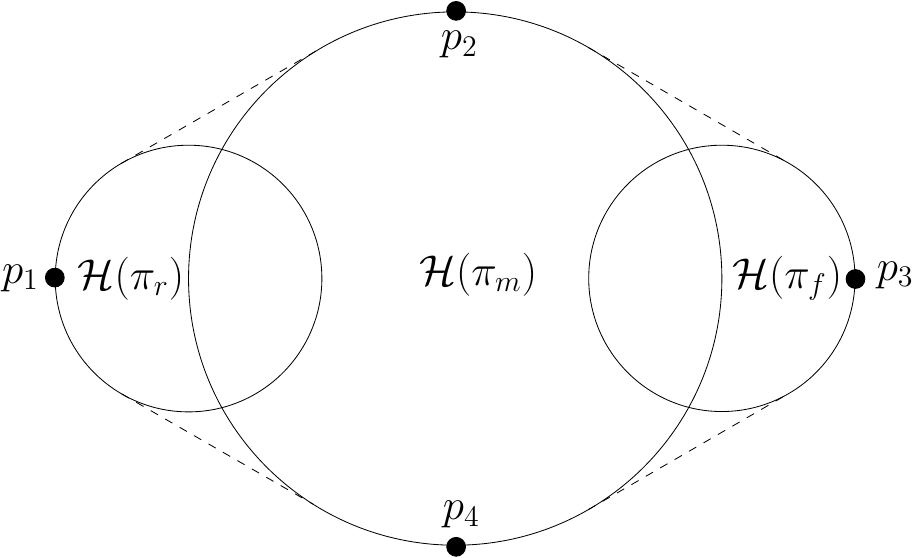}
    \caption{Illustrating the convex hull $\ch(\pi)$ in Case 6.}
    \label{fig:report3h}
\end{figure}

First, observe that $\ch(\pi_r)$ can have at most one maximal boundary portion appearing on the boundary of $\ch(\pi)$. To see this, $\ch(\pi)$ is the convex hull of $\ch(\pi_r)$ and $\ch(\pi_m\cup \pi_f)$. As~$\pi_r$ and $\pi_m\cup \pi_f$ are two consecutive subpaths of $\pi$, there are at most two common tangents between $\ch(\pi_r)$ and $\ch(\pi_m\cup \pi_f)$, implying that $\ch(\pi_r)$ has at most one maximal boundary portion appearing on $\partial \ch(\pi)$. Similarly, $\ch(\pi_f)$ has at most one maximal boundary portion appearing on $\partial \ch(\pi)$.

However, $\ch(\pi_m)$ can have at most two maximal boundary portions appearing on $\partial \ch(\pi)$. To see this, first of all, as argued above, $\ch(\pi_m)$ has at most one maximal boundary portion appearing on $\partial \ch(\pi_m\cup \pi_r)$. Note that $\ch(\pi)$ is the convex hull of $\ch(\pi_m\cup \pi_r)$ and $\ch(\pi_f)$. Since $\ch(\pi_f)$ has at most one maximal boundary portion appearing on $\partial \ch(\pi)$, we obtain that $\ch(\pi_m)$ have at most two maximal boundary portions appearing on $\partial \ch(\pi)$.

As such, the boundary of $\ch(\pi)$ has at most four maximal portions, each of which belongs to one of the three convex hulls of $S=\{\ch(\pi_f), \ch(\pi_m), \ch(\pi_f)\}$, and these maximal portions are connected by at most four edges of $\ch(\pi)$ each of which is a common tangent of two hulls of $S$ (see Figure~\ref{fig:report3h}). Suppose we know four vertices $p_i$, $1\leq i\leq 4$, one from each of these four boundary portions of $\ch(\pi)$. Then, we can construct $\ch(\pi)$ in $O(h+\log n)$ time by an algorithm using a walking procedure similar to that in Case 2, as follows.

Let $H_1$ be the convex hull of $S$ that contains $p_1$ as a vertex. Let $H_2$ and $H_3$ be the other two convex hulls of $S$. Starting from $p_1$, we perform a walking procedure by making reference to $\ch(\pi)$. More specifically, we walk clockwise on $\partial H_1$ with a step size $\log n$. Let $p_1'$ be the point on $\partial H_1$ such that there are $\log n$ vertices on $\partial H_1$ clockwise from $p_1$ to $p_1'$. We determine whether $p_1'$ is on $\partial \ch(\pi)$, which can be done in $O(\log n)$ time using the three trees $T_m$, $T_f$, and $T_r$ by Lemma~\ref{lem:helper-points}. If yes, then we continue the walk until after a step in which we either pass over $p_1$ or find that $p_1'$ is not on $\partial \ch(\pi)$. In the former case, we let $A_1$ be the set of all vertices of $H_1$. In the latter case, we make a counterclockwise walk on $\partial H_1$ starting from $p_1$ again and let $A_1$ be the subset of vertices that have been traversed before we stop the walking procedure.

We compute $A_i$ for $i=2,3,4$ similarly. Also, each $A_i$ can be sorted since all vertices of $A_i$ lie on the boundary of a single convex hull of $S$. As in the algorithm for Case 2, we have the following {\em key property}: All vertices of $\ch(\pi)$ are in $A=\bigcup_{1\leq i\leq 4} A_i$, and $A$ contains at most $O(\log n)$ points that are not on $\partial \ch(\pi)$. Therefore, $|A|=O(h+\log n)$ holds. As such, $\ch(\pi)$ is the convex hull of $A$ and can be computed in $O(h+\log n)$ time since $A$ can be sorted in $O(|A|)$ time and $|A|=O(h+\log n)$.

In light of the above algorithm, to report $\ch(\pi)$, it suffices to find the (at most) four points $p_i$, $1\leq i\leq 4$. In what follows, we first find two such points that are on the boundaries of two different convex hulls of $S$.
To this end, we first find a helper point $p_1\in \partial \ch(\pi_m\cup \pi_f)$ with $p_1\not\in \ch(\pi_r)$ and $p_1\in \partial \ch(\pi)$. Note that such a point $p_1$ always exists as $\ch(\pi_r)$ does not contain $\ch(\pi_m)$ (and thus $\ch(\pi_r)$ is not $\ch(\pi)$). Although we do not have a tree representing $\ch(\pi_m\cup \pi_f)$, we can still find such point $p_1$ in $O(\log n)$ time using Lemma~\ref{lem:helper-points}. The algorithm is the same as the one described in Case 4 for determining whether $\ch(\pi_r)$ contains $\ch(\pi_m\cup \pi_f)$ (in fact, the point $p$ found by that algorithm is our point $p_1$). Depending on whether~$p_1\in \partial \ch(\pi_m)$, there are two cases.

\begin{itemize}
    \item If $p_1\not\in \partial \ch(\pi_m)$, then $p_1\in \partial \ch(\pi_f)$. In this case, we find a helper point $p_2\in \partial \ch(\pi_m\cup \pi_r)$ with $p_2\not\in \ch(\pi_f)$ and $p_2\in \partial \ch(\pi)$. As above for $p_1$, $p_2$ can also be found in $O(\log n)$ time by slightly modifying the algorithm of Lemma~\ref{lem:helper-points}. Note that $p_2$ is either on $\partial \ch(\pi_m)$ or on~$\partial \ch(\pi_r)$. In either case, $p_2$ lies on a different convex hull than $p_1$ does.

    \item If $p_1\in \partial \ch(\pi_m)$, then we apply Lemma~\ref{lem:helper-points} on $\ch(\pi_m)$ and $\ch(\pi_f)$. Since $\ch(\pi_m)$ does not contain $\ch(\pi_f)$, the algorithm will find the extreme point $p_2'$ on $\partial \ch(\pi_f)$ along a certain direction $\rho$ such that $p_2'\not\in \ch(\pi_m)$. We next compute the extreme point $p_2''$ on $\partial \ch(\pi_r)$ along $\rho$. Let $p_2$ be the more extreme point among $p_2'$ and $p_2''$ along $\rho$. As such, $p_2$ must be on $\partial \ch(\pi)$. Since $p_2$ is either on $\partial \ch(\pi_f)$ or on $\partial \ch(\pi_r)$, $p_2$ is on a different convex hull than $p_1$ is.
\end{itemize}

The above finds two vertices $p_1, p_2 \in\partial \ch(\pi)$ on different convex hulls of $S$. Let $H_1$ and~$H_2$ be the two convex hulls containing $p_1$ and $p_2$ as vertices, respectively. Let $H_3$ be the third convex hull.

Using the walking procedure, for each $i=1,2$, we can compute a subset $A_i$ of vertices that contains $p_i$ such that either $A_i$ consists of all vertices of $H_i$, or $A_i$ contains the maximal boundary portion of $H_i$ appearing on $\partial \ch(\pi)$ that contains $p_i$. In either case, $A_i$ contains at most $O(\log n)$ vertices that are not on $\partial \ch(\pi)$. As such, $|A_1|,|A_2|=O(h+\log n)$. We compute the convex hull $H=\ch(A_1\cup A_2)$, which can be done in $O(h+\log n)$ time as $A_1\cup A_2$ can be sorted in linear time. Note that $H$ contains two {\em special edges}: Each of them connects a vertex of $A_1$ and a vertex of $A_2$, i.e., these are common tangents between the convex hulls $\ch(A_1)$ and $\ch(A_2)$; let $e_1$ and $e_2$ denote these two edges, respectively. Depending on whether~$H_3$ is $\ch(\pi_m)$, there are two cases to proceed.

\begin{itemize}
    \item If $H_3$ is $\ch(\pi_m)$, then we process a special edge $e_1$ as follows. Let $\rho$ be the normal of $e_1$ towards the outside of $H$. We compute the most extreme point $p_3$ of $H_3$ along $\rho$. We determine whether $p_3$ is inside $H$, which can be done in $O(h+\log n)$ time since $H$ is already computed and $|H|=O(h+\log n)$. If $p_3$ is inside $H$, then we ignore it. If $p_3$ is outside $H$, then $p_3$ must be on $\partial \ch(\pi)$. In this case, we compute the subset $A_3$ using the walking procedure, which takes $O(h+\log n)$ time.

    We process the other special edge $e_2$ similarly and obtain a point $p_4$ and the subset $A_4$, which takes $O(h+\log n)$ time. Finally, we compute the convex hull of $A=\bigcup_{1\leq i\leq 4}A_i$, which is $\ch(\pi)$. As discussed above, the total time is $O(h+\log n)$.

    \item If $H_3$ is not $\ch(\pi_m)$, then one of $H_1$ and $H_2$ is $\ch(\pi_m)$. Without loss of generality, we assume that $H_1$ is $\ch(\pi_m)$. In this case, we have computed one maximal portion of $\partial H_1$ that appears on $\partial \ch(\pi)$, which is contained in $A_1$. It is possible that there is another maximal portion of $\partial H_1$ appearing on $\partial \ch(\pi)$.

    Consider a special edge $e_1$. Define $\rho$ as the normal of $e_1$ toward the outside of $H$. We compute the most extreme point $p_3'$ of $H_3$ along $\rho$ and the most extreme point $p_3''$ of $H_1$ along $\rho$. Let $p_3$ be the more extreme point of $p_3'$ and $p_3''$ along $\rho$. By definition, $p_3$ must be on $\partial \ch(\pi)$.     
 
    We first consider the case where $p_3\not\in H$, implying that $e_1$ is not an edge of $\partial \ch(\pi)$. There are two cases depending on whether $p_3$ is $p_3'$ or $p_3''$.

    \begin{itemize}
        \item
                    If $p_3=p_3'$, then we have found a point $p_3\in \partial H_3$ that is on $\partial \ch(\pi)$. Then, we compute~$A_3$ using the walking procedure in $O(h+\log n)$ time. Next, we find a point~$p_4$ (if it exists) on $\partial H_1$. Let $\rho'$ be the normal of $e_2$ towards the outside of $H$. We compute the most extreme point $p_4$ of $H_1$ along $\rho'$. If $p_4$ is inside $H$ and then we ignore $p_4$; in this case, $\ch(\pi)$ is the convex hull of $A_1\cup A_2\cup A_3$, which can be computed in $O(h+\log n)$ time. If $p_4$ is outside $H$, then $p_4$ is on the second maximal portion of~$\partial H_1$ appearing on~$\partial \ch(\pi)$. In this case, we invoke the walking procedure with $p_4$ to compute $A_4$ and finally compute the convex hull of $A_1\cup A_2\cup A_3\cup A_4$, which is $\ch(\pi)$, in $O(h+\log n)$ time.

              \item
                    If $p_3=p_3''$, then $p_3$ is on the second maximal portion of $\partial H_1$ appearing on $\partial \ch(\pi)$. We first invoke the walking procedure on $p_3$ to compute $A_3$. Next, we consider the normal $\rho'$ of $e_2$ towards the outside of $H$. We compute the most extreme point $p_4$ of $H_3$ along $\rho'$. If $p_4$ is inside $H$ and then we ignore $p_4$; in this case, $\ch(\pi)$ is the convex hull of $A_1\cup A_2\cup A_3$, which can be computed in $O(h+\log n)$ time. If $p_4$ is outside $H$, then $p_4$ is on the maximal portion of $\partial H_3$ appearing on $\partial \ch(\pi)$. We invoke the walking procedure on $p_4$ to compute $A_4$ and finally compute the convex hull of $A_1\cup A_2\cup A_3\cup A_4$, which is $\ch(\pi)$, in $O(h+\log n)$ time.
          \end{itemize}    

\begin{figure}[t]
    \centering
    \includegraphics[height=2.0in]{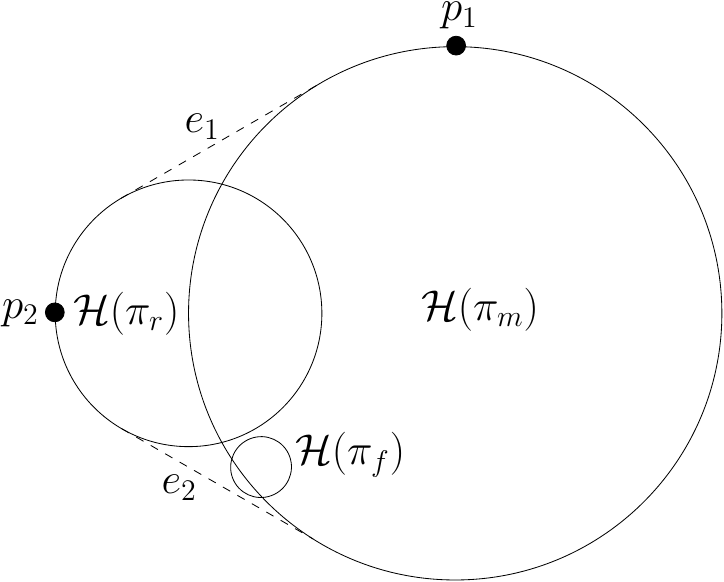}
    \caption{Illustrating the case where $\ch(\pi_f)\subseteq H$, with $H=\ch(\pi_r\cup \pi_m)$.}
    \label{fig:reportspecial}
\end{figure}
          
    The above is for the case where $p_3\not\in H$. If $p_3\in H$, then $e_1$ is an edge of $\partial \ch(\pi)$.
    In this case, we continue to process the other special edge $e_2$ in a way similar to the above for $e_1$. If $e_2$ is not an edge of $\partial \ch(\pi)$, then the algorithm finally computes $\ch(\pi)$ in $O(h+\log n)$ time as above. 
    Otherwise, we reach the case where both $e_1$ and $e_2$ are edges of $\partial \ch(\pi)$. This case can be easily handled as follows. First notice that $\ch(\pi)$ is the convex hull of $A_1\cup A_2\cup H_3$ (since otherwise, following the definitions of $A_1$ and $A_2$, either the above point $p_3$ would not be in $H$ or a symmetric situation would happen with respect to $e_2$).     
    Recall that $H_1$ is $\ch(\pi_m)$. Without loss of generality, we assume that $H_3$ is $\ch(\pi_f)$.       
    Note that $\ch(\pi_f)\subseteq H$ is possible (see Figure~\ref{fig:reportspecial}). If $\ch(\pi_f)\subseteq H$, then $\ch(\pi)$ is $H$; otherwise, the common tangents between $\ch(\pi_m)$ and $\ch(\pi_f)$ are also edges of $\partial \ch(\pi)$ (see Figure~\ref{fig:report3h}) and any point of $\partial\ch(\pi_f)\cap \partial \ch(\pi_f\cup \pi_m)$ is also on $\partial \ch(\pi)$. Our algorithm works as follows. We find a helper point $p_3\in \partial \ch(\pi_f)$ with $p_3\not\in \partial \ch(\pi_m)$, which can be done in $O(\log n)$ time using Lemma~\ref{lem:helper-points}. We check whether $p_3\in H$. If yes, then $\ch(\pi_f)\subseteq H$ and $H=\ch(\pi)$. As $H=\ch(A_1\cup A_2)$, $\ch(\pi)$ can be computed in $O(h+\log n)$ time in this case. If $p_3\not\in H$, then $\ch(\pi_f)\not\subseteq H$, and since $p_3\in \partial \ch(\pi_f)\cap  \partial \ch(\pi_f\cup \pi_m)$, $p_3$ is also on $\partial \ch(\pi)$ according to the above discussion.     
    As such, $p_3$ is on $\partial H_3$ and also on $\partial\ch(\pi)$. With $p_3$, we can now compute $A_3$ using the walking procedure in $O(h+\log n)$ time. Finally, we compute the convex hull of $A_1\cup A_2\cup A_3$, which is $\ch(\pi)$, in $O(h+\log n)$ time. 
    \end{itemize}

This completes the algorithm for reporting the convex hull $\ch(\pi)$. The runtime of the algorithm is $O(h+\log n)$.

\subparagraph{Remark.}
The above algorithm for reporting $\ch(\pi)$ relies on the three trees $T_f$, $T_m$, and $T_r$, representing $\ch(\pi_f)$, $\ch(\pi_m)$, and $\ch(\pi_r)$, respectively. Although the details would be tedious, it is not difficult to extend the algorithm to solve the four subpath case in which we only use the four trees $T_f$, $T_m^f$, $T_m^r$, and $T_r$, 
without using $T_m$. 
As such, our original data structure in Section~\ref{sec:ph-structure} is sufficient to perform \HR\ in $O(h+\log n)$ time.

\section{Lower bounds}
\label{sec:lower-bounds}

In this section, we prove lower bounds, which justify the optimality of our solutions to both the monotone path problem and the simple path problem.

The following theorem gives a lower bound for the 2-sided monotone path problem by a reduction from the set inclusion problem~\cite{Ben-Or83}. Note that the lower bound is even applicable to the amortized cost. It implies that our result in Theorem~\ref{thm:deque-convex-hull} is the best possible.

\begin{theorem}
    \label{thm:lower-bound-monotone}
    The following lower bounds hold for the 2-sided monotone path dynamic convex hull problem.
    \begin{enumerate}
        \item
              At least one of the operations  \textsc{InsertLeft}, \textsc{DeleteLeft}, \textsc{InsertRight}, \textsc{DeleteRight}, and \HR\ is required to take $\Omega(\log n)$ amortized time under the algebraic decision tree model, regardless of the value $h$, where $n$ is the size of the point set $P$ and $h=|\ch(P)|$.
        \item
              Similarly, at least one of the operations  \textsc{InsertLeft}, \textsc{DeleteLeft}, \textsc{InsertRight}, \textsc{DeleteRight}, and \SQ\ is required to take $\Omega(\log n)$ amortized time, where \SQ\ could be (but not limited to) any of the following queries: extreme point queries, tangent queries, bridge queries.
    \end{enumerate}
\end{theorem}
\begin{proof}
    We begin with proving the first theorem statement. The second one can be proven similarly.

    We use a reduction from the set inclusion problem. Let $A=\{a_1,a_2,\ldots,a_n\}$ and $B=\{b_1,b_2,\ldots,b_n\}$ be two sets of $n$ numbers each, such that the numbers of $B$ are distinct and sorted, i.e., $b_1<b_2<\cdots<b_n$. Without loss of generality, we assume that all numbers of $B$ are larger than $0$ and smaller than $1$. The {\em set inclusion} problem is to determine whether~$A\subseteq B$. The problem requires $\Omega(n\log n)$ time to solve under the algebraic decision tree model~\cite{Ben-Or83}.

    We create an instance of the 2-sided monotone path dynamic convex hull problem as follows. Define a function $f(x) = x (1 - x)$. For each number $b_i\in B$, we create a point $p_i=(b_i,f(b_i))$ in $\bbR^2$. Since $b_1<b_2<\cdots<b_n$, the points $p_i$'s in their index order are sorted from left to right. We show below that whether $A\subseteq B$ can be determined using $O(n)$ of the five operations in the first theorem statement, which will prove the first theorem statement.

    \begin{figure}[t]
        \centering
        \includegraphics[]{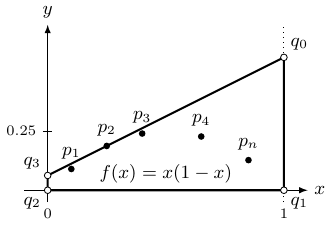}
        \caption{Reduction of set membership queries to 2-sided dynamic convex hull.}
        \label{fig:lower-bound}
    \end{figure}

    First, starting from $P=\emptyset$, we perform $n$ \textsc{InsertRight} operations for $p_1,p_2,\ldots,p_n$ in this order. Then, for each $a_i\in A$, we show that whether $a_i\in B$ can be determined using $O(1)$ operations. The function $f$ at the point $(x,f(x))$ has a tangent $\ell_x$ with slope $1-2x$ that intersects the vertical lines for $x=0$ and $x=1$ in the points $q_3(x)=(0,x^2)$ and $q_0(x)=(1,(1-x)^2)$. By definition, if $x$ is the $x$-coordinate of the point $p_i$ of $P$ (i.e., $x=b_i$), then $\ell_{x}$ is tangent to $\ch(P)$ at $p_i$ (see Figure~\ref{fig:lower-bound}, where the tangent goes through the point~$p_2$).

    Define $q_0=q_0(a_i)$, $q_1=(1,0)$, $q_2=(0,0)$, and $q_3=q_3(a_i)$. Note that all points of $P$ are below or on the segment $\overline{q_0q_3}$. We perform the following operations in order: \textsc{InsertRight}$(q_0)$, \textsc{InsertRight}$(q_1)$, \textsc{InsertLeft}$(q_2)$, \textsc{InsertLeft}$(q_3)$, and \HR. According to the above discussion, if $a_i\in B$, then $\ell_{a_i}$ must be tangent to $\ch(P)$ at the point~$p\in P$ whose $x$-coordinate is equal to $a_i$. This implies that the \HR\ operation will report the following five points: $q_i$, $0\leq i\leq 3$, and $p$ (we could resolve the collinear situation by a standard symbolic perturbation~\cite{ref:EdelsbrunnerSi90,ref:YapA90}). On the other hand, if $a_i\not\in B$, then all points of $P$ will be strictly below the line $\ell_{a_i}$ and thus only the four points $q_i$, $0\leq i\leq 3$, will be reported by \HR. As such, by checking the output of \HR, we are able to tell whether $a_i\in B$. If $a_i\not\in B$, then we know $A\not\subseteq B$. Otherwise, we continue to check the next number of $A$, but before doing so, we perform the following operations to delete the points~$q_i$'s: \textsc{DeleteLeft}$(q_3)$, \textsc{DeleteLeft}$(q_2)$, \textsc{DeleteRight}$(q_1)$, \textsc{DeleteRight}$(q_0)$.

    The above shows that whether $A\subseteq B$ can be determined using $O(n)$ of the four update operations and the \HR\ operation. This proves the first theorem statement.

    To prove the second theorem statement, we can follow the same reduction but replace the operation \HR\ with a corresponding query operation.

    Consider the extreme point queries. After we insert $q_i$, $0\leq i\leq 3$, as above, we perform an extreme point query with a direction $\rho$ being the normal of $\overline{q_0q_3}$. Then, $a_i\in B$ if and only if an extreme point of $\ch(P)$ returned by the query has $x$-coordinate equal to $a_i$. As such, whether $A\subseteq B$ can be determined using $O(n)$ of the four update operations as well as the extreme point queries.

    The tangent queries and the bridge queries can be argued similarly (e.g., by asking a tangent query at $q_0$ or a bridge query for the vertical line $x=a_i$). 
\end{proof}

For the simple path problem, we have the following lower bound, which is even applicable to the one-sided case in which insertions and deletions only happen at one end of the path (say, \IF\ and \DF). It implies that our result in Theorem~\ref{thm:dpch} is the best possible.

\begin{theorem}
    \label{thm:lower-bound-path}
    The following lower bounds hold for the one-sided simple path dynamic convex hull problem.
    \begin{enumerate}
        \item
              At least one of the operations  \IF, \DF, and \textsc{ReportHull} is required to take $\Omega(\log n)$ amortized time under the algebraic decision tree model, regardless of the value $h$, where $n$ is the number of the vertices of the simple path $\pi$ and $h=|\ch(\pi)|$.
        \item
              Similarly, at least one of the operations  \IF, \DF, and \SQ\ is required to take $\Omega(\log n)$ amortized time, where \SQ\ could be (but not limited to) any of the following queries: extreme point queries, tangent queries, bridge queries.
    \end{enumerate}
\end{theorem}
\begin{proof}
    We follow a similar proof to Theorem~\ref{thm:lower-bound-monotone} by a reduction from the set inclusion problem. Define $A$, $B$, and $p_i$, $1\leq i\leq n$, in the same way as in the proof of Theorem~\ref{thm:lower-bound-monotone}.

    Starting from $\pi=\emptyset$, we perform $n$ \IF\ operations for $p_1,p_2,\ldots,p_n$ in this order. Clearly, $\pi$ is a simple path. Then, for each $a_i\in A$, we can determine whether $a_i\in B$ using $O(1)$ operations as follows. Define $q_i$, $0\leq i\leq 3$, in the same way as before. We perform the following operations in order: \IF$(q_0)$, \IF$(q_1)$, \IF$(q_2)$, \IF$(q_3)$, and \HR. It can be verified that $\pi$ is still a simple path after~$q_i$'s, $0\leq i\leq 3$, are inserted. Following the same analysis as before, whether $a_i\in B$ can be determined based on the output of \HR. Afterward, we perform the \DF\ operations in the inverse order of the above insertions. As such, whether $A\subseteq B$ can be determined using $O(n)$ operations. This proves the first theorem statement.

    The second theorem statement can be proved similarly, following the same argument as in Theorem~\ref{thm:lower-bound-monotone}.
\end{proof}

\bibliographystyle{plainurl}
\bibliography{references}

\appendix



\end{document}